\documentclass[%
twocolumn,
superscriptaddress,
 amsmath,amssymb,
 prx,
showkeys,
showpacs,
]{revtex4-1}

\usepackage{graphicx}
\usepackage{dcolumn}
\usepackage{bm}
\usepackage{siunitx}
\usepackage{braket}
\usepackage{xcolor}
\usepackage{multirow}
\usepackage{array}
\usepackage{float}
\usepackage[title]{appendix}
\usepackage{etoolbox, chngcntr}
\usepackage{xr}
\usepackage{amsthm}
\usepackage{mathtools}
\usepackage{booktabs}
\usepackage[hidelinks]{hyperref}
\usepackage[capitalize]{cleveref}

\makeatletter
\def\maketitle{
\@author@finish
\title@column\titleblock@produce
\suppressfloats[t]}
\makeatother

\newcommand{\Rop}{\hat{R}}

\floatstyle{plaintop}
\restylefloat{table}

\begin{document}

\title{Above 99.9\,\% Fidelity Single-Qubit Gates, Two-Qubit Gates, and Readout in a Single Superconducting Quantum Device}
\author{Fabian Marxer}
\altaffiliation{These two authors contributed equally. fabian@meetiqm.com; jakub@meetiqm.com}

\author{Jakub Mro\.{z}ek}
\altaffiliation{These two authors contributed equally. fabian@meetiqm.com; jakub@meetiqm.com}

\author{Joona Andersson}

\author{Leonid Abdurakhimov}
\author{Janos Adam}
\author{Ville Bergholm}
\author{Rohit Beriwal}
\author{Chun Fai Chan}
\author{Saga Dahl}
\affiliation{IQM Quantum Computers, Espoo 02150, Finland}
\author{Soumya Ranjan Das}
\author{Frank Deppe}
\affiliation{IQM Quantum Computers, Munich 80992, Germany}
\author{Olexiy Fedorets}
\author{Zheming Gao}
\author{Alejandro Gomez Frieiro}
\affiliation{IQM Quantum Computers, Espoo 02150, Finland}
\author{Daria Gusenkova}
\affiliation{IQM Quantum Computers, Munich 80992, Germany}
\author{Andrew Guthrie}
\author{Tuukka Hiltunen}
\affiliation{IQM Quantum Computers, Espoo 02150, Finland}
\author{Hao Hsu}
\affiliation{IQM Quantum Computers, Munich 80992, Germany}
\author{Eric Hyypp\"a}
\author{Joni Ikonen}
\author{Sinan Inel}
\author{Shan~W.~Jolin}
\author{Azad Karis}
\author{Seung-Goo Kim}
\affiliation{IQM Quantum Computers, Espoo 02150, Finland}
\author{William Kindel}
\affiliation{IQM Quantum Computers, Munich 80992, Germany}
\author{Anton Komlev}
\author{Miikka Koistinen}
\author{Roope Kokkoniemi}
\author{Snigdha Kumar}
\affiliation{IQM Quantum Computers, Espoo 02150, Finland}
\author{Hsiang-Sheng Ku}
\author{Julia Lamprich}
\affiliation{IQM Quantum Computers, Munich 80992, Germany}
\author{Sami Laine}
\author{Alessandro Landra}
\author{Lan-Hsuan Lee}
\affiliation{IQM Quantum Computers, Espoo 02150, Finland}
\author{Nizar Lethif}
\affiliation{IQM Quantum Computers, Munich 80992, Germany}
\author{Per Liebermann}
\author{Wei Liu}
\author{Kunal Mitra}
\author{Tuomas Myll\"ari}
\author{Caspar Ockeloen-Korppi}
\author{Tuure Orell}
\author{Alexander Plyshch}
\author{Jukka R\"abin\"a}
\author{Arthur Rebello}
\affiliation{IQM Quantum Computers, Espoo 02150, Finland}
\author{Michael Renger}
\affiliation{IQM Quantum Computers, Munich 80992, Germany}
\author{Outi Reentil\"a}
\author{Jussi Ritvas}
\author{Sampo Saarinen}
\author{Otto Salmenkivi}
\author{Matthew Sarsby}
\author{Mykhailo Savytskyi}
\author{Ville Selinmaa}
\author{Matthew Steggles}
\author{Eelis Takala}
\author{Ivan Takmakov}
\author{Brian Tarasinski}
\author{Jani Tuorila}
\author{Alpo V\"alimaa}
\affiliation{IQM Quantum Computers, Espoo 02150, Finland}
\author{Jeroen Verjauw}
\affiliation{IQM Quantum Computers, Munich 80992, Germany}
\author{Jaap Wesdorp}
\affiliation{IQM Quantum Computers, Espoo 02150, Finland}
\author{Nicola Wurz}
\affiliation{IQM Quantum Computers, Munich 80992, Germany}
\author{Wei Qiu}
\author{Lihuang Zhu}
\affiliation{IQM Quantum Computers, Espoo 02150, Finland}

\author{Juha Hassel}
\affiliation{IQM Quantum Computers, Espoo 02150, Finland}
\author{Johannes Heinsoo}
\affiliation{IQM Quantum Computers, Espoo 02150, Finland}

\author{Attila Geresdi}
\affiliation{IQM Quantum Computers, Munich 80992, Germany}

\author{Antti Veps\"al\"ainen}
\email{avepsalainen@meetiqm.com}
\affiliation{IQM Quantum Computers, Espoo 02150, Finland}

\begin{abstract}
    Achieving high-fidelity single-qubit gates, two-qubit gates, and qubit readout is critical for building scalable, error-corrected quantum computers. However, device parameters that enhance one operation often degrade the others, making simultaneous optimization challenging. Here, we demonstrate that careful tuning of qubit–coupler coupling strengths in a superconducting circuit with two transmon qubits coupled via a tunable coupler enables high-fidelity single- and two-qubit gates, without compromising readout performance. As a result, we achieve a \SI{40}{h}-averaged CZ gate fidelity of \SI{99.93}{\percent}, simultaneous single-qubit gate fidelities of \SI{99.98}{\percent}, and readout fidelities over \SI{99.94}{\percent} in a single device. These results are enabled by optimized coupling parameters, an efficient CZ gate calibration experiment based on our new Phased-Averaged Leakage Error Amplification (PALEA) protocol, and a readout configuration compatible with high coherence qubits. Our results demonstrate a viable path toward scaling up superconducting quantum processors while maintaining consistently high fidelities across all core operations.
\end{abstract}
\date{\today}
\maketitle

\section{Introduction}

Scalable, fault-tolerant quantum computing requires low error rates across all core operations: single-qubit gates, two-qubit gates, and qubit readout. To enable resource-efficient implementation of logical operations, these error rates must fall well below the thresholds set by quantum error correction protocols~\cite{fowler2012surface, barends2014superconducting, googlequantumai2025quantum}. Among various quantum hardware platforms, superconducting qubits -- particularly transmons -- have emerged as a scalable and practical platform for quantum computing, owing to their favorable coherence properties, flexible control, and mature fabrication techniques~\cite{koch2007charge, arute2019quantum, kjaergaard2020superconducting}.

In such systems, single-qubit gate fidelities now routinely exceed \SI{99.9}{\percent}~\cite{googlequantumai2025quantum, kim2023evidence, li2023error, hyyppa2024reducing, rower2024suppressing}, and comparable fidelities have been achieved for two-qubit gates~\cite{ding2023high, li2024realization, lin2025days, zhang2024tunable}. This progress has been driven by recent advances in tunable coupler designs, such as grounded transmon couplers \cite{li2020tunable, xu2020high, collodo2020implementation, sung2021realization, ye2021realization, ding2023high}, floating transmon couplers \cite{sete2021floating, stehlik2021tunable, marxer2023long}, inductively mediated coupling schemes \cite{zhang2024tunable, lin2025days}, and double-transmon couplers \cite{li2024realization}, all of which enable dynamically controlled qubit–qubit interactions while suppressing residual coupling~\cite{yan2018tunable}.

At the same time, high-fidelity readout \cite{jeffrey2014fast, Walter2017rapid, chen2023high, swiadek2024enhancing, spring2025fast, kurilovich2025high} has become increasingly critical, not only for final measurements but also for mid-circuit measurements required in quantum error correction~\cite{googlequantumai2023suppressing, googlequantumai2025quantum, zhao2022realization} and feedback-based protocols~\cite{riste2012initialization, riste2012feedback, campagne2013persistent}. Advances such as shelving techniques~\cite{elder2020high, chen2023high, mallet2009single, wang2025longitudinal} and optimized readout chains incorporating traveling wave parametric amplifiers (TWPAs~\cite{yurke1989observation, yurke1996lownoise} have played a key role in pushing performance forward. Leveraging these developments, recent experiments have reported readout fidelities exceeding \SI{99.9}{\percent}~\cite{spring2025fast, kurilovich2025high}.

Despite these advances, simultaneously achieving high-fidelity single-qubit gates, two-qubit gates, and readout in a single device remains a major challenge. Device parameters optimized for one operation often degrade another. A key trade-off lies in the qubit–coupler coupling strength: strong coupling facilitates fast entangling gates, but also increases qubit-qubit hybridisation, leading to simultaneous single-qubit gate errors. Similarly, incorporating high-fidelity readout without degrading qubit coherence or gate performance requires careful architectural integration.

In this work, we demonstrate a holistic optimization strategy that addresses the trade-offs between single-qubit gates and two-qubit gates, while maintaining high-fidelity readout. Our device consists of two transmon qubits coupled via a tunable coupler. To guide the design, we simulate both single-qubit errors -- dominated by hybridisation-induced crosstalk and incoherent relaxation -- and two-qubit errors arising from leakage, incoherent processes, and residual swap interactions. By sweeping the qubit–coupler coupling strength, we identify a parameter regime that minimizes the combined gate error, which we then target in the device fabrication.
Furthermore, to enable precise calibration of the conditional-$Z$ (CZ) gate and to suppress coherent over- and under-rotation in the $\mathrm{span}( \{\ket{11}, \ket{02} \})$ subspace, we introduce a novel gate error characterization method, which we call the Phase-Averaged Leakage Error Amplification (PALEA) protocol. This method is specifically designed to coherently amplify CZ gate errors, especially their associated population leakage to the second excited state of the higher frequency qubit, while remaining robust against other leakage processes. Using PALEA to improve leakage calibration accuracy, we demonstrate that, using the same number of repetitions in the experiment, the leakage can be systematically reduced by at least a factor of two compared to standard leakage amplification methods~\cite{sung2021realization}, significantly enhancing compatibility with quantum error correction codes. Together with the qubit-coupler coupling optimization, we achieve a CZ gate fidelity of \SI{99.93}{\percent} and single-qubit gate fidelities above \SI{99.98}{\percent}.

Simultaneously, our readout architecture is designed to enable high-fidelity readout within a few hundred nanoseconds. Each qubit is coupled to its own readout resonator, which in turn is connected to an individual Purcell filter, not only to protect the qubit from Purcell decay, but also to suppress the off-resonant driving of untargeted readout resonators \cite{heinsoo2018rapid, ronkko2024onpremise}. The design parameters are chosen to achieve an optimal ratio of dispersive shift $\chi$ to resonator linewidth $\kappa$ (targeting $\chi/\kappa \approx 0.5$), enabling fast readout while maintaining high signal contrast \cite{Walter2017rapid}. Using a \SI{240}{\nano\second} readout pulse combined with shelving to the second excited state and a TWPA, we achieve simultaneous readout fidelities above \SI{99.94}{\percent} for both qubits in our setup~\cite{elder2020high, chen2023high, mallet2009single, wang2025longitudinal}. Using a pulse of the same duration but without the shelving gate, we additionally demonstrate a measurement with QNDness of \SI{99.3}{\percent} for both qubits, where QNDness quantifies how well the measurement preserves the qubit state across repeated measurements~\cite{lupascu2007quantum, hazra2025benchmarking, kurilovich2025high}.

Together, these results establish that high-fidelities across all core quantum operations are achievable. Notably, our approach is readily scalable to larger quantum processing units (QPUs), including the square-grid architecture necessary for surface code implementations \cite{arute2019quantum, krinner2022realizing, abdurakhimov2024technology}. Specifically, the qubit–coupler architecture demonstrated here -- where the qubits are coupled to a single tunable coupler -- can be naturally extended to a 2D grid by coupling each qubit to four independent couplers without significantly altering the coupling strengths or the anharmonicities of the qubits and couplers. This ensures that gate fidelity can in principle be preserved even as the system scales. Moreover, the chip layout offers sufficient space to integrate a dedicated Purcell filter for each readout resonator on the same layer as the qubits~\cite{ronkko2024onpremise, marxer2023long, abdurakhimov2024technology}, enabling scalable, high-fidelity readout without requiring multilayer fabrication or complex routing.

\section{Optimal Qubit-coupler coupling strength}

\begin{figure}[tb]
    \centering
    \includegraphics[width=1\columnwidth]{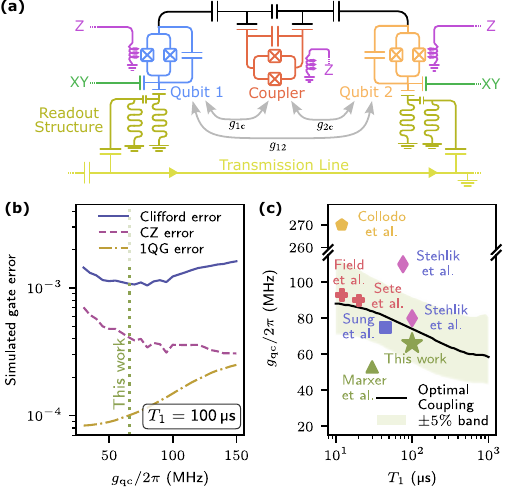}
    \caption{
        (a) Quasi-lumped-element circuit diagram of the device: two transmon qubits (blue and orange) are coupled via a central floating coupler qubit (red). Each qubit is equipped with its own $XY$ drive line (green), $Z$ flux line (violet), and a readout structure consisting of a readout resonator coupled to a Purcell filter (lime green). The readout structure is capacitively coupled to a common transmission line (yellow). The coupler qubit is flux-controlled but not driven with microwave pulses. Gray arrows indicate the effective coupling strengths between $\mathrm{Q}_1$ and the coupler ($g_\mathrm{1c}$), $\mathrm{Q}_2$ and the coupler ($g_\mathrm{2c}$), and the direct coupling between the two qubits ($g_\mathrm{12}$).
        (b) Simulated single-qubit gate, CZ gate, and two-qubit Clifford gate errors as a function of the qubit-coupler coupling strength $g_\mathrm{qc}/2\pi$, assuming an identical qubit-coupler coupling for both qubits, $g_\mathrm{qc} \coloneqq g_{1c} = g_{2c}$, and a fixed qubit relaxation time of $T_1 = \SI{100}{\micro\second}$. The vertical dotted line marks the coupling strength used in this work ($g_\mathrm{qc}/2\pi = \SI{66}{\mega\hertz}$).
        (c) Qubit-coupler coupling strengths, $g_\mathrm{qc}$, versus qubit relaxation times reported in selected previous works \cite{collodo2020implementation, sung2021realization, field2024modular, sete2021floating, stehlik2021tunable, marxer2023long}. The black line represents the simulated optimal coupling that minimizes the two-qubit Clifford error, and the shaded green region denotes the $\pm \SI{5}{\percent}$ band around this optimum.
    }
    \label{fig:fig_1}
\end{figure}

Our device consists of two flux-tunable floating transmon qubits coupled via a flux-tunable floating transmon coupler [Fig.~\ref{fig:fig_1}(a)], see Appendix~\ref{sec:appendix:experimental_setup} for the full experimental setup. The coupling structure is similar to the architecture introduced in Ref.~\cite{marxer2023long}. In this qubit--coupler--qubit topology, the couplings between Qubit~1 ($\mathrm{Q_1}$) and the coupler, $g_\mathrm{1c}$, Qubit~2 ($\mathrm{Q_2}$) and the coupler, $g_\mathrm{2c}$, as well as the direct qubit-qubit interaction $g_\mathrm{12}$ between $\mathrm{Q_1}$ and $\mathrm{Q_2}$ play a crucial role in the performance of both single- and two-qubit gates. For simplicity, we define a symmetric coupling \( g_\mathrm{qc} \coloneqq g_\mathrm{1c} = g_\mathrm{2c} \) for the numerical analysis.

While tunable couplers are usually employed to cancel the static $ZZ$ interaction between the computational states, the bare qubit states typically remain hybridised for most qubit–qubit detunings in this type of circuit~\cite{yan2018tunable, sung2021realization, valles2025optimizing}, see Appendix~\ref{subsec:appendix:zz_and_g}. This residual hybridisation gives rise to hybridisation crosstalk -- a coherent crosstalk mechanism in which driving one qubit unintentionally excites or induces phase shifts in other unwanted eigenstates, even in the absence of direct coupling between the driveline and those qubits. Consequently, minimizing the transverse coupling strength -- governed by $g_\mathrm{qc}$ and $g_{12}$ -- helps reduce such unwanted crosstalk, improving simultaneous single-qubit gate fidelities.

However, two-qubit gates benefit from a stronger $g_\mathrm{qc}$. First, higher couplings allow faster CZ gate operation, reducing incoherent error from relaxation and dephasing. Second, with stronger coupling, the same effective qubit-qubit interaction strength can be achieved with a smaller frequency excursion of the coupler. This arrangement reduces the hybridisation of the computational states with the coupler mode during the gate, thereby mitigating the impact of the coupler’s typically lower coherence on the gate fidelity~\cite{collodo2020implementation, sung2021realization, field2024modular}.

To explore this trade-off quantitatively, we simulate gate errors as a function of qubit-coupler coupling strength $g_\mathrm{qc}$, see Fig.~\ref{fig:fig_1}(b). The simulations assume qubit relaxation times of $T_1^\mathrm{Q1} = T_1^\mathrm{Q2} = \SI{100}{\micro\second}$, together with qubit and coupler frequencies and anharmonicities being similar to those of our fabricated device (see Appendix~\ref{sec:appendix:experimental_setup}). In the model, both qubits are assumed to have identical $T_1$ and pure dephasing times $T_\phi$. The coupler is assigned a relaxation time that is four times shorter than that of the qubits, reflecting its higher energy participation ratio (EPR) in lossy interfaces~\cite{minev2021energy}, which -- based on our present design -- is estimated to be four times larger. All components are modeled as coupled three-level Duffing oscillators \cite{koch2007charge}. Additionally, to isolate the impact of $g_\mathrm{qc}$, the direct qubit–qubit coupling $g_{12}$ is adjusted for each value of $g_\mathrm{qc}$ such that the coupler idling frequency (where $\zeta = 0$) remains fixed throughout the sweep.

For single-qubit gates, we use master equation simulations to evaluate simultaneous $\sqrt{X}$ operations on both qubits to capture both incoherent and coherent (hybridisation-induced) errors. We also simulate the CZ gate to extract leakage and incoherence contributions, see Appendix~\ref{sec:appendix:qubit-coupler_coupling} for further details. The resulting trends show that single-qubit errors increase with $g_\mathrm{qc}$, while CZ gate errors decrease. To evaluate overall gate performance, we calculate the estimated two-qubit Clifford gate error by weighting the single-qubit gate and two-qubit gate errors by their average prevalence in a random Clifford gate: 4.65 $\sqrt{X}$ gates and 1.5 CZ gates per Clifford~\cite{barends2014superconducting, mckay2017efficient}. These weighted error contributions yield the Clifford error curve shown in Fig.~\ref{fig:fig_1}(b).

From this analysis, we identify an optimal coupling strength near $g_\mathrm{qc}/(2\pi) = \SI{75}{\mega\hertz}$ that minimizes the total Clifford error. The fabricated device yields $g_\mathrm{1c}/(2\pi) = \SI{69}{\mega\hertz}$ and $g_\mathrm{2c}/(2\pi) = \SI{63}{\mega\hertz}$, corresponding to a geometric average $g_\mathrm{qc}/(2\pi) = \SI{66}{\mega\hertz}$, in close agreement with the design target.

To investigate the role of coherence, we repeat the simulation with $T_1$ values ranging from \SI{10}{\micro\second} to \SI{1}{\milli\second}. As shown in Fig.~\ref{fig:fig_1}(c), we observe that the optimal $g_\mathrm{qc}$ decreases with increasing $T_1$ due to the hybridisation errors of the single-qubit gates becoming the dominating error mechanism, see Appendix~\ref{sec:appendix:qubit-coupler_coupling} for further details.

We also compare our results to selected prior works that employ similar transmon-based qubit--coupler--qubit topology and report both $T_1$ and $g_\mathrm{qc}$~\cite{field2024modular, sete2021floating, sung2021realization, stehlik2021tunable, marxer2023long, collodo2020implementation}. Most devices operate within a $\pm \SI{5}{\percent}$ range of the simulated optimal coupling. Notably, the device in Ref.~\cite{collodo2020implementation} stands out as an exception, employing a significantly higher $g_\mathrm{qc}$ (by a factor of 3 to 4), possibly due to focusing solely on the CZ gate fidelity. We emphasize, however, that our simulation assumptions do not exactly match those of the compared works, as the qubit and coupler frequencies, anharmonicities, and coherence times differ from our model. Hence, this comparison should be regarded as qualitative.

Our findings provide guidance for future designs: as qubit coherence times increase beyond \SI{100}{\micro\second} \cite{tuokkola2025methods, bland20252d}, further reducing $g_\mathrm{qc}$ becomes beneficial. This not only lowers single-qubit and residual errors but also helps mitigate long-range hybridisation effects such as next-nearest-neighbor crosstalk -- an increasingly important factor in larger quantum processing units (QPUs) \cite{marxer2023long}.

\section{CZ gate calibration}
\label{sec:error_amplification_experiments}
\begin{figure}[tb]
    \centering
    \includegraphics[width=1\columnwidth]{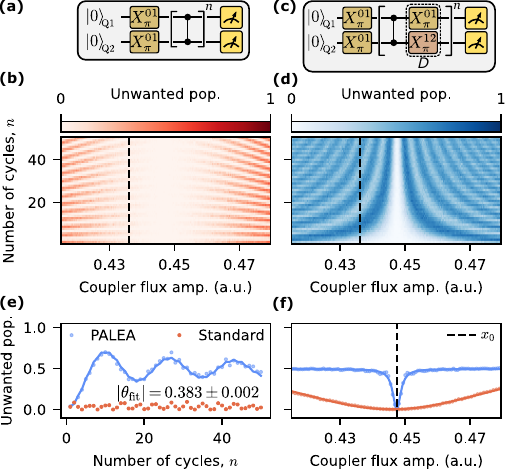}
    \caption{
        (a) Circuit diagram of the standard error amplification experiment. An $X_\pi^{01}$ pulse on each qubit excites the system to the $\ket{11}$ state, followed by repeated applications of the CZ gate.
        (b) Measured unwanted population in the $\ket{02}$ state as a function of the number of cycles and the coupler flux pulse amplitude.
        (c) Circuit diagram of the Phase-Averaged Leakage Error Amplification (PALEA) experiment. In each cycle, a CZ gate is followed by a dynamical decoupling layer $D$, consisting of an $X_\pi^{01}$ gate on the lower frequency qubit ($\mathrm{Q}_1$) and an $X_\pi^{12}$ gate on the higher frequency qubit ($\mathrm{Q}_2$).
        (d) Measured unwanted population as a function of the number of cycles and the coupler flux pulse amplitude. Here, population in the $\ket{02}$ ($\ket{11}$) state is considered unwanted (wanted) for even $n$, and population in $\ket{11}$ ($\ket{02}$) state is unwanted (wanted) for odd $n$.
        (e) Vertical slices from panels (b) and (d) at the flux amplitudes indicated by dashed lines. The solid curve shows a fit to the PALEA model [Eq.~\eqref{eq:I_n_npopex_main}], yielding an over-rotation angle $\left|\theta_\mathrm{fit}\right| = 0.383 \pm 0.002$.
        In contrast, the standard error amplification experiment signal shows low-contrast coherent oscillations at a frequency $2\mu$ such that $\theta < 2 \mu < \pi$.
        (f) Unwanted population, averaged over the number of cycles, as a function of the amplitude of the flux pulse. Solid lines represent Lorentzian fits used to extract the center offset $x_0$, indicated by the vertical dashed line. 
    }
    \label{fig:fig_2}
\end{figure}

In this work, the CZ gate is implemented using simultaneous flux pulses on both the coupler and $\mathrm{Q_1}$, with virtual $Z$ rotations applied to both qubits at the end of the gate to correct for dynamical phases. The flux pulse applied to the qubit consists of a constant plateau between cosine-shaped rising and falling edges to bring the $\ket{11}$ and $\ket{02}$ states rapidly close to resonance. The coupler flux pulse follows a Slepian-shaped envelope to minimize population leakage from $\mathrm{Q_1}$ to the coupler during the gate~\cite{martinis2014fast}. To account for distortions in the flux lines, we pre-compensate the pulses using exponential filters~\cite{rol2020time}.

Assuming a fixed gate duration, the calibration of this CZ gate implementation involves the minimization of three main imperfections: single-qubit $Z$ rotations, which we offset using virtual $Z$ gates~\cite{mckay2017efficient}; the $ZZ$ over-rotation angle, which is most sensitive to the frequency of $\mathrm{Q_1}$ during the gate \cite{marxer2023long}; and the population leakage to $\ket{02}$, sensitive to the frequency of the coupler during the gate. As design and coherence time improvements enable higher fidelity gates, the calibration techniques for all the necessary gate parameters need to improve in both precision and speed to keep the gate fidelity limited by incoherent errors. This can be achieved through methods based on phase estimation, employing error amplification with $n$ gates to approach the Heisenberg limit of error $\mathcal{O}(1/n)$, such as robust phase estimation~\cite{KimmelRPE2015, RudingerRPE2017}, Floquet calibration~\cite{arute2020observationseparateddynamicscharge}, or Matrix-Element Amplification using Dynamical Decoupling (MEADD)~\cite{Gross2024}.

Here, we use the MEADD protocol to minimize the error in two of the three CZ gate parameters: the $ZZ$ rotation angle, and the single-qubit virtual $Z$ rotations. Precise optimization of the third imperfection -- leakage to the $\ket{02}$, caused by the required diabatic transition between the $\ket{11}$ and $\ket{02}$ states -- is more challenging. Most existing error amplification protocols which could be used to minimize the $\ket{02}$ state leakage through extending them to the two-excitation subspace are either prone to interference by other error sources, slow, or incompatible with our hardware, see discussion in Appendix~\ref{sec:appendix:popex_dynamics}. To address this, we introduce a modified version of the MEADD protocol, called Phase-Averaged Leakage Error Amplification (PALEA), that employs higher excited state dynamics and phase averaging to isolate the leakage parameter, eliminating the need to control or calibrate the qubit drive phase at each repetition and thereby simplifying and accelerating the implementation in our architecture.

\subsection{Leakage Error Amplification}
The diabatic CZ gate between transmon qubits with tunable couplers is frequently described in terms of tuning the states $\ket{11}$ and $\ket{02}$ to resonance, and then activating their coupling for a fixed time~\cite{abad2025impactofdecoherence}. In practice, however, these two levels are not exactly on resonance, and thus this model is often insufficient. Here, we model an imperfect CZ gate as an arbitrary two-level unitary $\hat{U}_g \in \mathrm{SU}(2)$ acting within the $\mathrm{span}( \{\ket{11}, \ket{02} \})$ subspace in the laboratory reference frame (see discussion in Appendix~\ref{subsec:appendix:lab_frame}). Ignoring the common phase acquired by both the $\ket{11}$ and $\ket{02}$ states, this unitary can be decomposed as $\Rop_z(\alpha)\Rop_x(\theta)\Rop_z(\beta)$, where $\Rop_j (\phi) = \exp(-i \frac{\phi}2 \hat{\sigma}_j)$ denotes a rotation by angle $\phi$ around axis $j$, and ${\sigma}_j$ is the corresponding Pauli operator in the $\{\ket{11}, \ket{02}\}$ basis~\cite{Nielsen_Chuang_2010}. The resulting action on the $\ket{11}$ and $\ket{02}$ states is thus given by:

\begin{equation}\label{eq:U_g_definition}
\begin{aligned}
\hat{U}_g\ket{11} = e^{i \frac{\alpha + \beta}2} \cos (\theta /2) \ket{11} + ie^{i \frac{\alpha - \beta}2} \sin (\theta /2) \ket{02} & \\
\hat{U}_g\ket{02} = ie^{-i \frac{\alpha - \beta}2} \sin (\theta /2) \ket{11} - e^{-i \frac{\alpha + \beta}2} \cos (\theta /2) \ket{02}
\end{aligned}
\end{equation}

The angle of interest, $\theta$, governs the population transfer between the $\ket{11}$ and $\ket{02}$ states and is primarily controlled by the coupler frequency -- and thus by the coupler flux pulse amplitude. Any deviation from $\theta=0$ indicates an error. Hence, the aim of a leakage amplification experiment is to precisely measure the over-rotation angle $\theta$. In contrast, the angles $\alpha$ and $\beta$, describing the difference of phases acquired by states $\ket{11}$ and $\ket{02}$,  are irrelevant, as the $\ket{02}$ state lies outside the computational basis. 

The basic leakage amplification experiment~\cite{sung2021realization} repeatedly applies CZ gate candidates to an initial $\ket{11}$ state, see Fig.~\ref{fig:fig_2}(a). As shown in Fig.~\ref{fig:fig_2}(b), sweeping both the number of CZ gates and the coupler flux pulse amplitude while measuring the population of the $\ket{02}$ state (noted as unwanted population in Fig.~\ref{fig:fig_2}(b)) reveals low contrast oscillations when the coupler flux pulse amplitude significantly deviates from the optimal value of \SI{0.447}{}. These oscillation frequencies cannot easily be used to characterize the leakage, as they depend on the phases of the eigenvalues of $\hat{U}_g$, not just $\theta$, see discussion in Appendix~\ref{sec:appendix:popex_dynamics}. In addition, this experiment is susceptible to an accidental resonance with a leakage mode, for example to the coupler, appearing as vertical lines of elevated unwanted population~\cite{sung2021realization}. Using Slepian-shaped coupler flux pulses reduces the leakage to the coupler and thus the presence of the associated spurious lines~\cite{martinis2014fast, sung2021realization}, though residual effects may still interfere with the calibration of the leakage to $\ket{02}$, see Appendix \ref{sec:appendix:out_of_subspace_dynamics}.
\subsection{PALEA protocol}
To avoid the aforementioned issues in amplifying the over-rotation angle $\theta$, we introduce the PALEA protocol. A dynamical decoupling (DD) layer $D$ is added after each CZ gate, consisting of an $X_\pi^{12}(\phi_1)$ gate with phase $\phi_1$, acting on the $\mathrm{span}( \{\ket{1}, \ket{2} \})$ subspace of the higher frequency transmon and a phased-$X_\pi^{01}(\phi_0)$ gate with phase $\phi_0$ on the lower frequency transmon, see Fig.~\ref{fig:fig_2}(c). Hence, $D$ acts on the $\mathrm{span}( \{\ket{11}, \ket{02} \})$ subspace as a $\pi$ rotation: 
\begin{equation}\label{eq:palea_dd_definition}
    D = \Rop_z(\phi_0 - \phi_1) \Rop_x(\pi)\Rop_z(\phi_1 - \phi_0).
\end{equation}
This is essentially an adaptation of the SWAP error protocol from Ref.~\cite{Gross2024} (Eqs. 55-60), adjusting the DD layer to act on the $\mathrm{span}( \{\ket{11}, \ket{02} \})$ subspace instead of the $\mathrm{span}( \{\ket{10}, \ket{01} \})$ subspace. Importantly, to correctly implement the $D$ operation in the laboratory frame, special care needs to be taken to ensure that the phases of the subsequent $D$ layers within the same gate sequence are properly accounted for, see Appendix~\ref{subsec:appendix:palea_phase_tracking}.

After $n$ cycles of CZ gates followed by the DD layers the probability of measuring $\ket{11}$ is:
\begin{equation}\label{eq:npopex_expectation}
    V_n(\theta, \Delta \phi) = \left|\braket{11 | \left[\Rop_z(\Delta \phi) \Rop_x(\pi - \theta )\right]^n | 11}  \right|^2,
\end{equation}
where $\Delta \phi = 2\phi_0 - 2\phi_1  - \alpha + \beta$ depends both on the difference of the phases of the applied $X_\pi^{01}$ and $X_\pi^{12}$ gates and the intrinsic $Z$ rotations induced by each CZ gate between the $\ket{11}$ and $\ket{02}$ states, see discussion in Appendix~\ref{sec:appendix:npopex_model}. Due to the $D$ layer acting as a $\pi$ rotation between the states $\ket{11}$ and $\ket{02}$, the cycle unitary $\Rop_z(\Delta \phi) \Rop_x(\pi - \theta )$ is now dominated by the $X$ rotation rather than the $Z$ rotation for small $\theta$ values, greatly increasing visibility in the final $Z$ measurement. In contrast to the MEADD protocol, which evaluates the dynamics for only two specific values of $\Delta \phi$, PALEA averages over a wide range of random angles. This averaging is introduced either passively -- through the inability to reset the phase difference between the $X_\pi^{01}$ and $X_\pi^{12}$ gates between repetitions, as is common for systems using CZ-based gate sets -- or it can be implemented actively by manually varying the phase difference in each repetition. When sampling a broad range of $\Delta \phi$ values the expected value of the measured $\ket{11}$ state population is given by
\begin{equation}\label{eq:I_n_def_V_n}
I_n(\pi - \theta) = \frac{1}{2\pi} \int_{0}^{2\pi} \mathrm{d} \phi \ V_n(\theta, \phi).
\end{equation}

The purpose of the averaging is to decouple the unimportant angles $\alpha$ and $\beta$ from the target over-rotation angle $\theta$ and perform only one averaged measurement.

The integral in Eq.~\eqref{eq:I_n_def_V_n} can be solved analytically (see Appendix~\ref{subsec:appendix:npopex_echo_solving}), yielding a closed-form expression for the signal in the PALEA experiment:
{\small
\begin{equation}\label{eq:I_n_npopex_main}
\begin{split}
I_n(\pi -\theta) = 1 - \cos^2(\theta/2) \sum_{m = 0}^{n - 1} (-1)^m P_m(\cos \theta),
\end{split}
\end{equation}
}
where $P_m$ is the $m$-th Legendre polynomial. The model predicts a high-contrast signal without requiring additional calibration -- adaptive or otherwise -- of any other parameter, allowing for an accurate estimation of $\theta$ with fewer measurement repetitions. Eq.~\eqref{eq:I_n_npopex_main} can be quickly evaluated using the Clenshaw recursion algorithm for Legendre series, see Appendix~\ref{subsec:appendix:npopex_echo_understanding}.

\subsection{Calibration using PALEA}
By fitting the model to experimental data, we can accurately fit the over-rotation angle $\theta$, and hence, the leakage to the $\ket{02}$ state. For example, the signal $I_n$ at a coupler flux pulse amplitude of \SI{0.436}{} [Fig.~\ref{fig:fig_2}(d)] can be fitted with the model in Eq.~\eqref{eq:I_n_npopex_main} to extract $\left|\theta_\mathrm{fit}\right| = 0.383 \pm 0.002$, corresponding to a population transfer to the $\ket{02}$ state of $\sin^2 (\theta_\mathrm{fit}/2) \approx \SI{3.6}{\percent}$ per CZ gate,  see Fig.~\ref{fig:fig_2}(e). Conversely, the standard error amplification experiment yields a low-contrast signal for the coupler flux pulse amplitude of \SI{0.436}{} [Fig.~\ref{fig:fig_2}(b)], oscillating at a faster frequency $2\mu = 2\cos^{-1}\left(\cos (\theta/2) \cos [(\alpha + \beta)/2] \right)$ (Appendix~\ref{sec:appendix:popex_dynamics}).

As shown in Fig.~\ref{fig:fig_2}(f), averaging the PALEA signal over the number of cycles yields a sharp dip which can be accurately fitted with a simple Lorentzian model to determine the optimal coupler flux pulse amplitude $x_0$. This enables fast and efficient calibration of a high fidelity CZ gate. In contrast, averaging the signal of the standard error amplification experiment produces a broader dip, reducing the precision of identifying the optimal coupler flux pulse amplitude and potentially increasing the leakage to the $\ket{02}$ state -- unless a significantly higher number of repetitions is used, see Appendix~\ref{subsec:appendix:npopex_echo_understanding}.

The decoupling of the over-rotation angle $\theta$ from the $Z$ rotation angles allows using PALEA as part of an optimal quantum control~\cite{Koch2022} optimization protocol varying a large number of parameters of exotic pulse shapes, from Nelder-Mead to Bayesian optimization with Gaussian processes~\cite{jones1998efficient, wang2018batched}, without the need to find additional $Z$ rotation angles, greatly increasing its utility. Additionally, including the dynamical decoupling $D$ to target a specific pair of states and averaging over the phases of the involved $X$ rotations allows one to isolate the amplified exchange process between these states -- here $\ket{11}$ and $\ket{02}$. The other processes which may involve the state $\ket{11}$, but are not explicitly targeted by $D$ -- such as leakage to the coupler -- are effectively converted to incoherent processes and thus not amplified, see discussion in Appendix~\ref{subsec:appendix:measuring_coupler_leakage}. This removes a possible systematic bias introduced to the calibration by coupler leakage processes being randomly amplified at certain coupler pulse amplitudes, see details in Appendix~\ref{subsec:appendix:npopex_echo_out_of_subspace}.

\section{Optimizing and Benchmarking the CZ Gate}
\label{sec:cal_opt_and_benchmark_cz_gates}

\begin{figure}[tb]
    \centering
    \includegraphics[width=1\columnwidth]{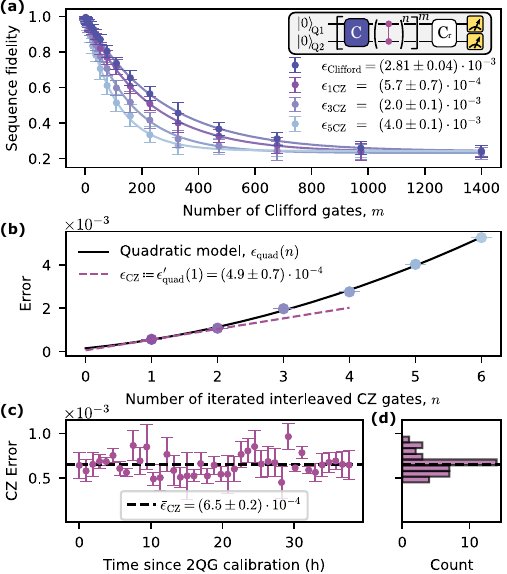}
    \caption{
        (a) Iterative interleaved randomized benchmarking results for one, three, and five iteratively interleaved CZ gates, where the error bars indicate the standard deviation. The extracted errors per $n$-CZ gates, $\epsilon_{n\mathrm{CZ}}$, are shown with fit uncertainties. Inset: Iterative IRB circuit, consisting of a Clifford gate $\mathrm{C}$ followed by $n$ CZ gates, repeated $m$ times, and finalized with an inverting Clifford gate $\mathrm{C_r}$.
        (b) Extracted error versus number of iteratively interleaved CZ gates. The black line shows a quadratic fit $\epsilon_\mathrm{quad}(n)$; the slope at $n=1$ gives the CZ gate error $\epsilon_\mathrm{CZ} = \SI{4.9(0.7)e-4}{}$ (dashed line).
        (c) Time series of $\epsilon_\mathrm{CZ}$ over \SI{40}{\hour}. Error bars indicate the standard error from the quadratic fit. The average CZ gate error, $\bar{\epsilon}_\mathrm{CZ} = \SI{6.5(0.2)e-4}{}$, is shown as a dashed line.
        (d) Histogram of $\epsilon\mathrm{CZ}$ values from (c), with the same average indicated.
    }
    \label{fig:fig_3}
\end{figure}

To accurately characterize gate fidelity, we use iterative interleaved randomized benchmarking (IRB)~\cite{sheldon2016characterizing, marxer2023long, lin2025days}. In this protocol, $n$ CZ gates are interleaved between Clifford gates instead of just one. Fig.~\ref{fig:fig_3}(a) illustrates the protocol for $n = 1, 3, 5$. The extracted total error after applying $n$ interleaved CZ gates, $\epsilon(n)$, is shown in Fig.~\ref{fig:fig_3}(b). We fit this dependence to a quadratic model,
\[
\epsilon_\mathrm{quad}(n) = a n^2 + b n + c
\]
to disentangle coherent and incoherent error sources: incoherent errors increase linearly with $n$, while coherent errors and errors with long temporal correlations accumulate quadratically~\cite{sheldon2016characterizing}. A further advantage of iterative IRB is the explicit extraction of the constant offset term $c$, which may arise from residual non-Markovian effects -- such as flux or drive pulse distortions from the surrounding Clifford gates~\cite{hyyppa2024reducing} or coherent flux noise~\cite{brillant2025randomized} -- unrelated to the interleaved CZ gates. In standard IRB, this offset is indistinguishable from the CZ gate error, leading to systematic over- or underestimation, whereas iterative IRB separates it from the intrinsic gate error.

We define the CZ gate error as the slope of the curve at $n=1$, i.e., $\epsilon_\mathrm{CZ} \coloneqq \epsilon'_\mathrm{quad}(1) = 2a + b$~\cite{xiong2025scalable, renger2025superconducting}. Unlike $\epsilon_\mathrm{quad}(1)$, which resembles the gate error definition of standard IRB, $\epsilon'_\mathrm{quad}(1)$ is independent of $c$ while still capturing the coherent error of a single CZ gate. Although coherent error contribution grows with $n$, in many applications such residual errors can be converted into incoherent errors using protocols such as randomized compiling~\cite{hashim2021randomizedcompiling}, ensuring that the coherent errors do not add up constructively.

To minimize coherent CZ gate errors, we calibrate both the coupler and qubit flux pulse amplitudes using PALEA, discussed in Section~\ref{sec:error_amplification_experiments}. In addition, we optimize the gate durations by sweeping the coupler and qubit flux pulse widths separately. Full re-calibration and iterative IRB with three-state discrimination is performed for each duration setting to benchmark both the CZ gate error and the leakage to the $\ket{2}$ states (see Appendix~\ref{sec:appendix:cz_gate_optimization}). We find the optimal CZ gate performance at a total gate duration of \SI{33}{\nano\second}, consisting of a \SI{22}{\nano\second} coupler pulse overlapping with a \SI{27}{\nano\second} qubit pulse and \SI{3}{\nano\second} buffers on each side.

After calibration and optimization of the CZ gate, we benchmark it using iterative IRB over a continuous period of \SI{40}{\hour}, re-calibrating only the single-qubit gates between measurements [Fig.~\ref{fig:fig_3}(c)]. For each run of the iterative IRB experiment, we use a fixed seed to generate the set of Clifford sequences, and find an average CZ error of
\[
\bar{\epsilon}_\mathrm{CZ} = \SI{6.5(0.2)e-4}{},
\]
where the uncertainty is the standard error of the mean. The lowest observed CZ gate error during the measurement series was $\epsilon_\mathrm{CZ}^\mathrm{min} = \SI{4.9(0.7)e-4}{}$. The histogram of data points in Fig.~\ref{fig:fig_3}(d) confirms a nearly Gaussian distribution with standard deviation $\sigma(\bar{\epsilon}_\mathrm{CZ}) = \SI{1.1e-4}{}$.

The corresponding average Clifford gate error is measured to be
\[
\bar{\epsilon}_\mathrm{Clifford} = \SI{2.79(0.05)e-3}{},
\]
in good agreement with the estimate based on the Clifford gate decomposition
\[
\epsilon^\mathrm{estimate}_\mathrm{Clifford} = 1 - \left( (1-\epsilon_\mathrm{CZ})^{1.5} \cdot (1-\epsilon_{\sqrt{X}})^{4.65} \right)  = \SI{2.4e-3}{},
\]
using $\epsilon_{\sqrt{X}} \approx \SI{3e-4}{}$. Note that in these experiments, the single-qubit gate duration was set to \SI{40}{\nano\second}, which had not yet been fully optimized. The improved fidelities presented later in the paper use shorter, optimized gate durations.

We also compare the iterative IRB result to the CZ error estimated using standard IRB analysis, which yields
\[
\bar{\epsilon}^\mathrm{IRB}_\mathrm{CZ} = \SI{7.3(0.2)e-4}{}.
\]
This value is significantly higher than the gate error extracted from iterative IRB. In our data, the quadratic fit reveals an average offset of $c = 1.1 \times 10^{-4}$, likely from residual non-Markovian effects unrelated to the CZ gate. Subtracting this offset from the standard IRB result yields $\bar{\epsilon}^\mathrm{IRB}_\mathrm{CZ} - c = \SI{6.2e-4}{}$, in close agreement with the iterative IRB value, indicating that standard IRB mistakenly attributes the offset $c$ to the CZ gate error.

\section{Error budget of the full system}

\begin{figure*}[tb]
    \centering
    \includegraphics[width=1\textwidth]{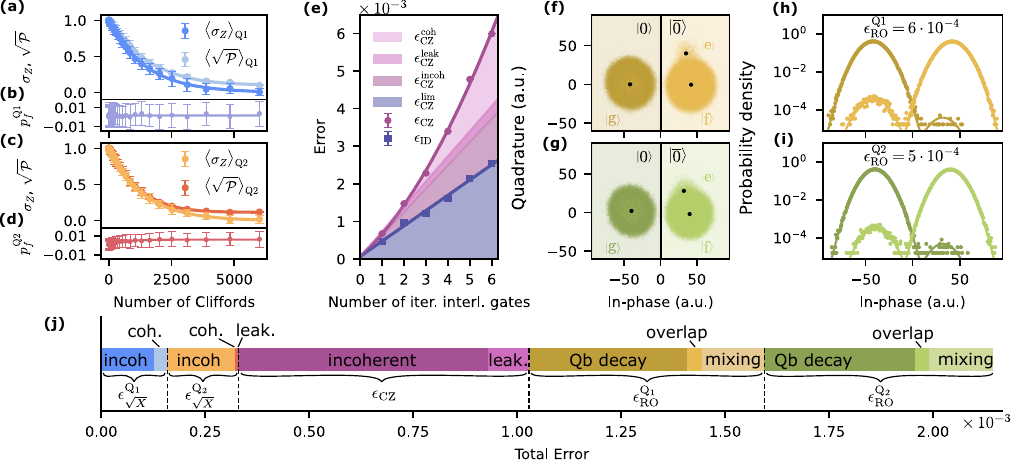}
    \caption{
        (a,b) Purity randomized benchmarking (RB) results for $\mathrm{Q}_1$. Shown are the average $\langle \sigma_Z \rangle$, the sequence purity $\langle \sqrt{\mathcal{P}} \rangle$, and the leakage population $p_f$ as a function of the number of Clifford gates. The extracted $\sqrt{X}$ gate error is $\epsilon_{\sqrt{X}}^{\mathrm{Q}1} = \SI{1.59(0.02)e-4}{}$.
        (c,d) Same as (a,b), but for $\mathrm{Q}_2$ yielding $\epsilon_{\sqrt{X}}^{\mathrm{Q}2} = \SI{1.68(0.02)e-4}{}$.
        (e) Error budget for the CZ gate extracted from leakage interleaved RB. Errors for CZ (purple) and identity gates (blue) are shown as a function of the number of gate repetitions. CZ data is fitted with a quadratic model, identity with a linear model. The blue shaded region indicates the CZ error limit $\epsilon_\mathrm{CZ}^\mathrm{lim}$ in the absence of coupler-induced effects. Contributions from incoherent errors ($\epsilon_\mathrm{CZ}^\mathrm{incoh}$), leakage to the second excited states ($\epsilon_\mathrm{CZ}^\mathrm{leak}$), and coherent errors ($\epsilon_\mathrm{CZ}^\mathrm{coh}$) are shown in shades of purple.
        (f,g) Single-shot readout clouds for $\mathrm{Q}_1$ and $\mathrm{Q}_2$, based on two million repetitions. Using a shelving technique via an $X_\pi^{12}$ pulse, $\ket{e}$ state populations are moved mostly to $\ket{f}$ state prior to readout. Black lines indicate thresholds between $\ket{0}$ and $\ket{\bar{0}}$.
        (h,i) Probability density versus in-phase readout signals used to extract readout fidelities: $\epsilon_\mathrm{RO}^{\mathrm{Q}1} = \SI{5.7(0.2)e-4}{}$ and $\epsilon_\mathrm{RO}^{\mathrm{Q}2} = \SI{5.4(0.2)e-4}{}$.
        (j) Summary of combined error contributions for each qubit, separated into single-qubit gate errors (incoherent, coherent, leakage), two-qubit gate errors (incoherent, leakage), and readout errors (decay, overlap, mixing). The total average error across the system is $\SI{2.1e-3}{}$.
    } 
    \label{fig:fig_4}
\end{figure*}

To evaluate the overall performance of our quantum device, we construct a comprehensive error budget encompassing all core operations: single-qubit gates, two-qubit (CZ) gates, and readout. While the CZ gate fidelity is a key benchmark for entangling operations, achieving consistently low errors across all layers of the control stack is critical for scalability and fault-tolerant performance. This section details how each component is individually optimized, calibrated, and characterized.

\subsection{Single-qubit gates}
\label{subsec:single_qubit_gates}

Single-qubit gates are implemented using cosine-shaped envelopes combined with the derivative removal by adiabatic gate (DRAG) technique~\cite{motzoi2009simple, chen2016measuring} to suppress leakage to the second excited state. The quadrature envelope is optimized for leakage minimization, and the resulting phase errors are corrected using virtual $Z$ gates~\cite{mckay2017efficient}. Our native gate set consists of $\sqrt{X}$ rotations.

To quantify the single-qubit gate fidelity, we use purity randomized benchmarking (RB) with three-state discrimination to capture incoherent, coherent, and leakage errors~\cite{wallman2015estimating, feng2016estimating, wood2018quantification, hyyppa2024reducing}. We sweep the gate duration to identify the optimal operating point, finding a minimum error at \SI{26}{\nano\second}. This relatively long duration, compared to recent literature~\cite{hyyppa2024reducing, werninghaus2021leakage}, is primarily constrained by compression effects in the up-conversion mixer, which limit the usable drive power.

Figures~\ref{fig:fig_4}(a)--(b) show the purity RB results for $\mathrm{Q_1}$, categorizing single-qubit gate errors into incoherent, coherent, and leakage errors. We extract a single-qubit gate error of $\epsilon_{\sqrt{X}}^{\mathrm{Q1}} = \SI{1.59(0.02)e-4}{}$, with the dominant contribution from incoherent error, $\epsilon_{\sqrt{X}}^{\mathrm{Q1,\, incoh}} = \SI{1.26(0.01)e-4}{}$. The residual error consists of coherent calibration errors such as control pulse distortion, and a negligible amount of leakage errors to the $\ket{2}$, remaining below the measurable threshold of $\SI{1e-6}{}$. Figures~\ref{fig:fig_4}(c)--(d) show analogous results for $\mathrm{Q_2}$, where we find $\epsilon_{\sqrt{X}}^{\mathrm{Q2}} = \SI{1.68(0.02)e-4}{}$, composed of $\epsilon_{\sqrt{X}}^{\mathrm{Q2,\, incoh}} = \SI{1.62(0.01)e-4}{}$, $\epsilon_{\sqrt{X}}^{\mathrm{Q2,\, coh}} = \SI{4(2)e-6}{}$, and leakage $\epsilon_{\sqrt{X}}^{\mathrm{Q2,\, leak}} = \SI{3.7(0.5)e-6}{}$.

The low leakage errors are expected given the relatively long gate durations used and the additional suppression from the DRAG pulse. Simulations for a $\sqrt{X}$ gate with a duration of \SI{26}{\nano\second}, qubit-qubit detuning of \SI{110}{\mega\hertz}, and qubit relaxation times of $T_1 \approx \SI{100}{\micro\second}$ predict a gate error of $\epsilon_{\sqrt{X}}^{\mathrm{sim}} = \SI{1.3e-4}{}$ for both qubits, fully limited by incoherent error, agreeing well with the measured values. At these gate durations, detuning, and optimized qubit–coupler couplings, the expected hybridisation error is below \SI{1e-5}{}, and therefore does not limit the performance. This is experimentally confirmed by comparing gate fidelities of the two qubits individually, where we observe nearly identical performance indicating negligible hybridisation error. These results also imply that the microwave crosstalk, measured to be approximately \SI{-20}{\decibel}, is not a limiting factor in our setup at the given qubit-qubit detuning.

We note that the simulated $\sqrt{X}$ gate errors in Fig.~\ref{fig:fig_1}(b) are based on optimally short gate durations (on the order of \SI{15}{\nano\second}), while our experiments use longer pulses due to limitations mentioned above. Hence, the simulated gate error reaches $\SI{1e-4}{}$, even when averaging over non-ideal qubit-qubit detunings.

\subsection{CZ gate error decomposition}
\label{subsec:cz_error_decomposition}

To assess the two-qubit gate performance, we use the iterative IRB method described in Section~\ref{sec:cal_opt_and_benchmark_cz_gates}. In addition to measuring the CZ gate error directly, we seek to disentangle its error contributions.

We first determine the coherence-limited error floor of the qubits by applying iterative IRB with an interleaved identity (ID) gate of the same duration as the CZ gate (\SI{33}{\nano\second}). The extracted error per ID gate is $\epsilon_\mathrm{ID} = \SI{4.1(0.2)e-4}{}$ [Fig.~\ref{fig:fig_4}(e)], which serves as an estimate for the lower bound of CZ gate error in the absence of coupler-induced decoherence, $\epsilon_\mathrm{CZ}^\mathrm{lim}$.

Next, we estimate the incoherent component of the CZ gate error, $\epsilon_\mathrm{CZ}^\mathrm{incoh}$, by measuring $T_1$ and $T_2$ of the hybridised qubit modes as a function of the coupler flux bias, and integrating over the flux pulse waveform to compute effective coherence times (see Appendix~\ref{sec:appendix:coherence_limit_cz_gate}). This yields $\epsilon_\mathrm{CZ}^\mathrm{incoh} = \SI{6.2e-4}{}$. The increased value compared to $\epsilon_\mathrm{ID}$ confirms the impact of the coupler’s shorter relaxation time ($T_1^\mathrm{TC} = \SI{33}{\micro\second}$) versus the qubit lifetimes ($T_1^\mathrm{Q} \approx \SI{100}{\micro\second}$).

To characterize leakage errors to higher qubit excited states, we implement leakage randomized benchmarking (LRB)~\cite{wood2018quantification, Rol2019Leakage} using three-state discrimination of the iterative IRB data. We extract a total leakage error to the second excited states of both qubits, $L_\mathrm{CZ}^\mathrm{tot} = \SI{1e-4}{}$. Most leakage originates from $\mathrm{Q_2}$ ($L_\mathrm{CZ}^\mathrm{Q2} = \SI{9e-5}{}$), which transiently populates the $\ket{2}$ state during the CZ gate.

Summing the incoherent and leakage contributions gives $\epsilon_\mathrm{CZ}^\mathrm{incoh} + L_\mathrm{CZ}^\mathrm{tot} = \SI{7.2e-4}{}$, which closely matches the measured average CZ error of $\bar{\epsilon}_\mathrm{CZ} = \SI{6.5e-4}{}$, suggesting that most error contributions are captured by this decomposition. The residual error is attributed to coherent error $\epsilon_\mathrm{CZ}^\mathrm{coh}$, which manifests as a quadratic contribution in iterative IRB and is likely caused by flux pulse distortions that induce inter-gate memory effects.

We also characterize the SWAP and coupler leakage errors and demonstrate that neither limits the performance of our CZ gate. The SWAP error arises from coherent population exchange between the $\ket{10}$ and $\ket{01}$ states and is described by a two-level unitary operator $\hat{U}_s$ acting in the $\mathrm{span}( \{\ket{10}, \ket{01} \})$ subspace. Analogous to the leakage analysis in Section~\ref{sec:error_amplification_experiments}, $\hat{U}_s$ is decomposed as $\Rop_z(\alpha_s)\Rop_x(\theta_s)\Rop_z(\beta_s)$. Here, $\theta_s$ quantifies the SWAP error, while $\alpha_s$ and $\beta_s$ correspond to conditional and single-qubit phases that can be calibrated separately. Using PALEA, we measure $\theta_s = \SI{0.0198(0.0006)}{\radian}$, which corresponds to a SWAP error of $\epsilon_\mathrm{CZ}^\mathrm{SWAP} = \SI{3.9(0.2)e-5}{}$, confirming that SWAP errors are not a limiting factor for the CZ gate fidelity (see Appendix~\ref{subsec:appendix:swap_errors} for further discussion).

Leakage to the coupler is more challenging to quantify, as the coupler lacks direct $XY$ control and readout. The absence of a dedicated drive and measurement line precludes the use of PALEA or randomized benchmarking techniques to assess population transfer to the coupler excited state. Instead, we adopt a modified Floquet calibration method~\cite{arute2020observationseparateddynamicscharge}, in which a variable delay $\tau$ is inserted between repeated CZ gates. During this delay, the state $\ket{101}$ and the leaked state $\ket{011}$ acquire a relative phase under the free Hamiltonian evolution of the system. For specific $\tau$ values, this phase evolution constructively amplifies any residual coupler population, enabling its detection. We measure a leakage-induced rotation angle of $\theta_\mathrm{TC} = \SI{0.041(0.001)}{\radian} $ per CZ gate, corresponding to a coupler leakage error of $\epsilon_\mathrm{CZ}^\mathrm{TC} =\frac45 - \frac45 \cos^4 \left(\theta_{\mathrm{TC}}/4 \right) = \SI{1.7(0.1)e-4}{}$, where we have modeled the coupler leakage as an incoherent decay of $\mathrm{Q}_1$. This decay is already partially accounted for by the difference between $\epsilon_\mathrm{ID}$ and $\epsilon_\mathrm{CZ}$ (see Appendix~\ref{subsec:appendix:measuring_coupler_leakage} and~\ref{subsec:appendix:coupler_fidelity_contribution}). We attribute the small coupler leakage primarily to the use of a Slepian-shaped envelope for the coupler flux pulse, which effectively minimizes leakage during the CZ gate.

\subsection{Shelved readout and fidelity analysis}

\begin{table*}[tb]
    \centering
    \caption{Readout performance metrics for shelved and QND-optimized readout.}
    \begin{tabular}{|l|c|c|c|c|}
        \hline
        \multirow{2}{*}{} & \multicolumn{2}{c|}{\(\mathrm{Q}_1\)} & \multicolumn{2}{c|}{\(\mathrm{Q}_2\)} \\ \cline{2-5}
         & Shelved & QND-opt. & Shelved & QND-opt. \\ \hline
        Readout error, \(\epsilon_\mathrm{RO}\) 
        & \SI{5.7(0.2)e-4}{} & \SI{1.92(0.03)e-3}{} 
        & \SI{5.4(0.2)e-4}{} & \SI{1.59(0.03)e-3}{} \\
        QND error, \( 1 - \mathcal{Q}\) 
        & \SI{1.06(0.02)e-2}{} & \SI{6.7(0.2)e-3}{} 
        & \SI{1.11(0.02)e-2}{} & \SI{5.9(0.2)e-3}{} \\
        Readout-induced leakage, L 
        & \SI{1.005(0.006)e-2}{} & \SI{2.1(0.2)e-5}{} 
        & \SI{1.05(0.005)e-2}{} & \SI{6.6(0.3)e-4}{} \\ \hline
    \end{tabular}
    \label{tab:readout}
\end{table*}

For high-fidelity measurement, we use a shelving-based readout technique: an $X_\pi^{12}$ gate is applied just prior to a rectangular readout pulse to transfer the $\ket{e}$ state to $\ket{f}$, thereby slightly increasing its effective readout lifetime~\cite{elder2020high, chen2023high, mallet2009single, wang2025longitudinal}. We begin by optimizing the readout frequency, followed by a joint optimization of the readout amplitude and duration for maximum readout fidelity. Finally, after optimizing the integration weights, the configuration consists of a \SI{40}{\nano\second} shelving gate followed by a \SI{240}{\nano\second} readout pulse, resulting in a total duration of \SI{280}{\nano\second}.

Figures~\ref{fig:fig_4}(f)--(g) show single-shot in-phase and quadrature distributions for $\mathrm{Q_1}$ and $\mathrm{Q_2}$ based on two million samples. $\ket{g}$ state preparations are mostly classified as $\ket{0}$, with rare $\ket{\bar{0}}$ misclassifications mostly due to measurement-induced excitations. $\ket{e}$ state preparations -- shelved to $\ket{f}$ with occasional $f \!\to\! e$ relaxation -- are classified as $\ket{\bar{0}}$; rare $e \!\to\! g$ relaxation yields $\ket{0}$ misclassifications. Counting the assignment errors, we extract readout errors of $\epsilon_\mathrm{shelved}^{\mathrm{Q1}} = \SI{5.7(0.2)e-4}{}$ and $\epsilon_\mathrm{shelved}^{\mathrm{Q2}} = \SI{5.4(0.2)e-4}{}$, with uncertainties obtained via bootstrapping \cite{efron1979bootstrap}.

From the histogram fits shown in Figs.~\ref{fig:fig_4}(h)--(i), we identify the main readout error sources for $\mathrm{Q}_1$ as decay error $\epsilon_\mathrm{shelved}^{\mathrm{Q1,decay}} = \SI{3.8(0.2)e-4}{}$, mixing error $\epsilon_\mathrm{shelved}^{\mathrm{Q1,mix}} = \SI{1.48(0.01)e-4}{}$, and negligible overlap error $\epsilon_\mathrm{shelved}^{\mathrm{Q1,overlap}} = \SI{3.6(0.1)e-5}{}$. For $\mathrm{Q}_2$, the error profile follows the same pattern, with decay error $\epsilon_\mathrm{shelved}^{\mathrm{Q2,decay}} = \SI{3.6(0.2)e-4}{}$ as the leading contribution, while mixing and overlap errors remain at levels comparable to those of $\mathrm{Q}_1$.

Although shelving increases the effective relaxation time during measurement, qubit decay remains the dominant limitation to the readout fidelity. Potential improvements include shelving to the third excited state, which further extends the effective lifetime of the readout manifold.

While shelving provides high-fidelity readout suitable for final measurements, it is less practical for mid-circuit measurements required in quantum error correction~\cite{googlequantumai2025quantum} due to the potentially increased risk of leakage to the $\ket{f}$ state. Compatibility with such measurements can be restored if high-fidelity reset from the $\ket{f}$ state is available~\cite{magnard2018fast, sunada2022fastreadout, zhou2021rapid, ding2025multipurpose, kim2025fastunconditional}, enabling the qubit to be returned to its pre-measurement state when needed. 

Ideally, mid-circuit measurements are implemented as quantum non-demolition (QND) readouts, which preserve the qubit state regardless of the measurement outcome~\cite{lupascu2007quantum}. In practice, no readout is perfectly QND, and the goal is to maximize its QNDness $\mathcal{Q}$, defined as~\cite{hazra2025benchmarking}:
\begin{equation}
    \mathcal{Q} = \left( \mathcal{Q}_\mathrm{g} + \mathcal{Q}_\mathrm{e} \right) /2,
\end{equation}
with
\begin{equation}
    \begin{split}
        \mathcal{Q}_\mathrm{g} &= P(\mathrm{g}, 0|\mathrm{g}) + P(\mathrm{g}, 1|\mathrm{g}) \\
        \mathcal{Q}_\mathrm{e} &= P(\mathrm{e}, 0|\mathrm{e}) + P(\mathrm{e}, 1|\mathrm{e}),
    \end{split}
\end{equation}
where $P(a, b|c)$ denotes the probability that the qubit ends up in state $a$ and is assigned outcome $b$, given the qubit is initially in state $c$. We adapt the readout-induced leakage benchmarking (RILB) protocol~\cite{hazra2025benchmarking} to extract $\mathcal{Q}$; see Appendix~\ref{sec:appendix:qndness} for details on the experiment and models.

Using the fidelity-optimized shelved readout parameters, and appending an additional $X_\pi^{12}$ gate after the readout pulse to return population from $\ket{f}$ to $\ket{e}$, we observe that the QNDness is limited by elevated leakage rates (Table~\ref{tab:readout}). This leakage is attributed to the decay from $\ket{f}$ to $\ket{e}$ during the readout pulse, followed by re-excitation to $\ket{f}$ via the second $X_\pi^{12}$ gate.

As an alternative, we optimize a standard (unshelved) readout pulse for QNDness, which typically involves lowering the readout amplitude to reduce the number of photons in the resonator and minimize excitation to higher states~\cite{sank2016measurementinduced, khezri2023measurementinduced, hazra2025benchmarking}. The QND-optimized readout reduces leakage by up to two orders of magnitude relative to the shelved readout and improves the QNDness up to $\mathcal{Q} =\SI{99.41}{\percent}$, albeit with reduced fidelity, see Table~\ref{tab:readout}.

These results place our QND performance alongside other high-fidelity QND measurements reported for QPU-compatible architectures~\cite{pereira2023parallel, gard2024fasthighfidelity, sunada2022fastreadout, dassonneville2020fasthighfidelity, mori2025highpowerreadout, kurilovich2025high}.

\subsection{Total system-level error budget}

Combining all contributions, we define the total system-level error as
\[
\epsilon_\mathrm{tot} = \sum_{i=1,2} \left( \epsilon_{\sqrt{X}}^{\mathrm{Q}i} + \epsilon_\mathrm{RO}^{\mathrm{Q}i} \right) + \epsilon_\mathrm{CZ}  \approx \SI{2.1e-3}{}.
\]
Depending on the application, different weights may be assigned to each component. For instance, quantum error correction demands low leakage and coherent errors to ensure decoder compatibility~\cite{googlequantumai2025quantum}, while dynamical decoupling schemes would utilize more single-qubit gates, hence increasing their weight. Note that our simple definition of the system error, $\epsilon_\mathrm{tot}$, does not account for all error mechanisms -- such as idling errors -- which are relevant for quantum error correction.

Our results show that by carefully optimizing each subsystem -- through pulse shaping, error-amplification, and shelving-based readout -- it is possible to approach the fidelity regime needed for scalable fault-tolerant quantum computing.

\section{Conclusions}

We have demonstrated simultaneous high-fidelity single-qubit gates, CZ gates, and readout in a superconducting quantum device composed of two transmon qubits coupled via a tunable floating coupler. Through targeted device-level simulations, we identified the optimal qubit–coupler coupling strength that minimizes the total Clifford error by balancing hybridisation-induced crosstalk and incoherent gate errors. This design strategy was validated experimentally, yielding single-qubit gate fidelities above \SI{99.98}{\percent}, a 40-hour-averaged CZ gate fidelity of \SI{99.93}{\percent}, and readout fidelities exceeding \SI{99.94}{\percent}.

A key enabler of this performance is our newly developed PALEA (Phase-Averaged Leakage Error Amplification) protocol, which allows precise calibration of coherent over- and under-rotation errors in the $\mathrm{span}( \{\ket{11}, \ket{02} \})$ subspace. Combined with iterative interleaved randomized benchmarking, PALEA enables both the suppression and quantification of coherent errors, leading to a factor of two reduction in leakage and an accurate estimation of the residual CZ gate error.

We also presented a detailed error budget that quantifies the incoherent, coherent, and leakage components of both single- and two-qubit gates, as well as readout. The readout fidelity was achieved via shelving to the second excited state and was broken down into decay, mixing, and overlap contributions. In addition, we optimized non-shelved readout for QNDness, obtaining values exceeding \SI{99.3}{\percent} for both qubits. The total system-level error is estimated to be $\SI{2.1e-3}{}$, approaching operation errors required to efficiently scale up quantum error correction codes and thus representing a key advance towards scalable, fault-tolerant quantum computing.

Importantly, we demonstrated a device architecture and control methodology compatible with scalable layouts. Our design allows each qubit to be coupled to multiple tunable couplers while preserving gate parameters, and the chip layout supports dedicated Purcell filters for each qubit on a single fabrication layer. These features make our approach well-suited for extension to 2D square-lattice architectures required, for example, for surface code implementations.

Together, our results provide a blueprint for achieving uniformly high fidelities across all core quantum operations in a scalable superconducting platform, bringing fault-tolerant quantum computing closer to reality.

\section{Acknowledgements}
We acknowledge Ali Yurtalan, Lucas Ortega, Pavel Titov for supporting the construction and maintenance of the experimental setup, Balint Csatari, Rakhim Davletkaliyev, Guillermo Alonso, Jyrgen Luus, and Umut Deniz \"Ozugurel for additional software support. We would additionally like to thank the rest of the IQM team for creating the entire infrastructure, laying the foundation of this work. We would like to thank Sumeru Hazra and Wei Dai for discussions on the RILB protocol.

A.VE. conceptualized the project, F.M. and J.M. planned and executed the experiments and analyzed the experimental data, A.G., A.L., E.T., T.M., J.R., conducted and analyzed the microwave simulations, J.AN., S.L., T.O., M.ST., H.H., O.S., and J.T. provided theoretical modeling and simulations, A.G., C.O.-K., A.L., and S.R.D. designed and simulated the sample, F.M., J.M., J.AD., M.ST., M.SAV., R.B., C.F.C., O.F., Z.G., A.G.F., D.G., T.H., E.H., J.I., S.I., S.J., M.K., J.L., N.L., P.L., J.R., B.T., J.W., A.VA., V.B., and A.VE. developed the experiment and analysis software, 
A.P., S.-G.K, W.L., W.Q., A.KO., L.Z., H.-H.L. and A.KA. fabricated, coordinated, and selected the device, K.M., V.S., S.D. bonded and packaged the device, R.K., M.SAR., and M.SAV. maintained and improved the experimental setup, L.A., S.K., A.R., S.S., and M.ST. worked on coherence improvements for this sample, W.K., H.-S.K., F.D., J.HA., M.R., and J.V. provided background support ,F.M., J.M., J.AN., J.AD., J.HE., and A.VE. drafted or revised the manuscript, and J.HE., A.G., and A.VE. supervised the work.

This work was supported by Business Finland through projects QuTi (40755/31/2020) and CfoQ (787/31/2025). Parts of this work are included in patent applications filed by IQM Finland Oy. 

\bibliography{ref}
\pagebreak
\clearpage

\onecolumngrid
\title{Supplementary Material: Above 99.9\,\% Fidelity Single-Qubit Gates, Two-Qubit Gates, and Readout in a Single Superconducting Quantum Device}

\maketitle

\clearpage

\setcounter{figure}{0}
\setcounter{table}{0}
\setcounter{page}{1}
\setcounter{section}{0}

\newcolumntype{C}{>{\centering\arraybackslash} m{1.4cm} }

\renewcommand{\thefigure}{S\arabic{figure}}
\renewcommand{\theHfigure}{S\arabic{figure}}
\renewcommand{\thetable}{S\arabic{table}}
\renewcommand{\theHtable}{S\arabic{table}}
\renewcommand{\bibnumfmt}[1]{[S#1]}

\newcommand{\note}[1]
  {\begingroup{\color{blue}[NOTE: \textit{#1}]}\endgroup}
\newtheorem{lemma}{Lemma}[section]

\appendix
\setcounter{equation}{0}
\renewcommand{\theequation}{S\arabic{equation}}
\makeatletter
\@removefromreset{equation}{section}
\makeatother

\section{EXPERIMENTAL SETUP}
\label{sec:appendix:experimental_setup}

\begin{figure}[tb]
    \centering
    \includegraphics[width=1\columnwidth]{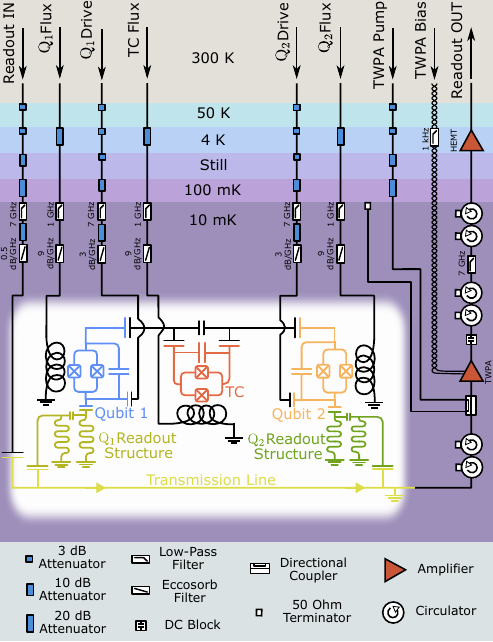}
    \caption{
        Schematic of the experimental setup. The background colors depict the different temperature stages of the dilution refrigerator. The coaxial microwave cables (straight black lines) are used to control and read out the qubit and coupler states. The three-wave–mixing traveling-wave parametric amplifier (TWPA) is biased via a twisted-pair cable and pumped through a coaxial microwave line. The schematic of the sample in the white box consists of two flux-tunable qubits (blue and orange), a readout structure for each qubit (lime green), and a flux-tunable coupler (red).
    }
    \label{fig:fig_sm_exp_setup}
\end{figure}

The experiments are performed in a cryogenic dilution refrigerator with a base temperature of approximately \SI{10}{\milli\kelvin}. A schematic of the full wiring configuration is shown in Fig.~\ref{fig:fig_sm_exp_setup}. The sample consists of two frequency-tunable floating transmon qubits coupled via a frequency-tunable floating transmon coupler. A summary of the key device parameters is provided in Table~\ref{tab:parameters}.

Each qubit has an individual microwave drive line, equipped with \SI{39}{\decibel} of attenuation distributed across the various temperature stages to minimize thermal noise. Additional low-pass and Eccosorb filters are installed at the mixing chamber stage to suppress spurious high-frequency signals. The Eccosorb filters add approximately another \SI{15}{\decibel} attenuation at the qubit frequencies.

The qubits are driven using a Zurich Instruments HDAWG arbitrary waveform generator (AWG), which produces intermediate-frequency pulses that are upconverted to the qubit frequencies using IQ mixers. Dedicated HDAWG channels are also used to control the voltage on the flux line of $\mathrm{Q}_1$ and the coupler. For $\mathrm{Q}_2$, which requires only static flux control, we use a Qblox SPI D5a voltage source module to deliver DC biases.

The readout bus consists of a transmission line capacitively coupled to the readout structure of each qubit. Each readout structure comprises a dedicated resonator–Purcell filter pair, designed to suppress spontaneous emission into the environment while maintaining the coupling strength required for fast, high-fidelity measurement. Readout pulses are generated using a Zurich Instruments UHFQA, with frequency multiplexing applied to address both qubits simultaneously. The UHFQA output is upconverted and sent through the dilution refrigerator input chain.

The readout signal is first amplified by a three-wave–mixing traveling-wave parametric amplifier (TWPA) at the \SI{10}{\milli\kelvin} stage and subsequently by a \SI{4}{\kelvin} high-electron-mobility transistor (HEMT) amplifier. The amplified signal is then downconverted at room temperature and digitized by the UHFQA. The TWPA is pumped by a Rhode \& Schwarz microwave generator at \SI{11}{\giga\hertz} and biased using two DC output channels of a Qblox SPI D5a voltage source.

For reference, Fig.~\ref{fig:fig_sm_exp_setup} includes the location and labeling of all attenuators, filters, amplifiers, and circulators, as well as the thermal stage at which each component is mounted. All flux and drive lines are thermally anchored at each temperature stage to minimize heat load and thermal noise coupling to the qubits.

\begin{table*}[tb]
    \centering
    \renewcommand{\arraystretch}{1.4}
    \begin{tabular}{|l| C C C|}
        \hline
         Quantity, symbol (unit) & Qubit 1 & Qubit 2 & Coupler  \\ \hline
        Readout resonator frequency, $\omega_\mathrm{R}/2\pi$ (GHz)& 6.190 & 6.350 & -\\
        Readout Purcell filter frequency, $\omega_\mathrm{Rpf}/2\pi$ (GHz) & 6.187 & 6.329 & -\\
        Effective readout resonator bandwidth, $\kappa_\mathrm{eff} / 2\pi$ (MHz) & 6.1 & 3.4 & -\\
        Readout circuit dispersive shift, $\chi_\mathrm{R} / 2 \pi$ (MHz) & 2.5 & 2.6 & -\\
        Qubit frequency at sweet spot, $\omega_\mathrm{ge}/2\pi$ (GHz) & 4.799 & 4.910 & 4.850\\
        Qubit frequency at idling, $\omega_\mathrm{ge}/2\pi$ (GHz) & 4.799 & 4.910 & 3.672\\
        Transmon anharmonicity, $\alpha$ (MHz) & $-191$ & $-191$ & -216\\
        Energy relaxation time at idling, $T_1$ (µs) & 86 & 102 & 33\\
        Transverse relaxation time at idling, $T_2$ (µs) & 74 & 63 & 1.4\\
        Echoed transverse relaxation time at idling, $T_2^\mathrm{echo}$ (µs) & 140 & 104 & 3.1\\
        Pure dephasing time at idling, $T_\phi$ (µs) & 130 & 91 & 1.4\\
        Qubit-coupler coupling, $g_\mathrm{qc} / 2\pi$ (MHz) & 69.6 & 62.2 & -\\
    \hline
    \end{tabular}
    \caption{Summary of device parameters.}
    \label{tab:parameters}
\end{table*}

\section{QUBIT-COUPLER COUPLING SIMULATION}
\label{sec:appendix:qubit-coupler_coupling}

\begin{figure*}[tb]
    \centering
    \includegraphics[width=1\textwidth]{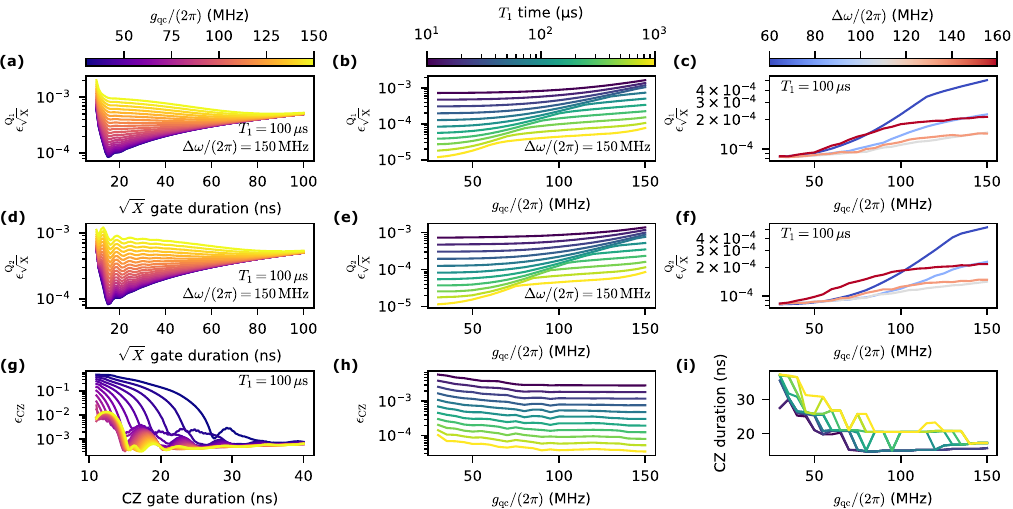}
    \caption{
        (a), (d) Simulated $\sqrt{X}$ gate error vs gate duration for $\mathrm{Q}_1$ and $\mathrm{Q}_2$, respectively, at $T_1 = \SI{100}{\micro\second}$, $\Delta\omega = \SI{150}{\mega\hertz}$, and qubit–coupler couplings from \SI{30}{\mega\hertz} to \SI{150}{\mega\hertz}.
        (b), (e) Simulated $\sqrt{X}$ gate error vs qubit–coupler coupling for a range of $T_1$ times (\SI{10}{\micro\second}–\SI{1}{\milli\second}) at fixed $\Delta\omega / (2\pi) = \SI{150}{\mega\hertz}$, and an optimized gate duration.
        (c), (f) Simulated $\sqrt{X}$ gate error vs qubit–coupler coupling at $T_1 = \SI{100}{\micro\second}$ for five different qubit–qubit detunings $\Delta\omega / (2\pi) \in \{\SI{60}, \SI{85}, \SI{110}, \SI{135}, \SI{160}\}~\si{\mega\hertz}$.
        (g) Simulated CZ gate error vs gate duration for different qubit–coupler couplings at $T_1 = \SI{100}{\micro\second}$.
        (h) Simulated CZ gate error vs qubit–coupler coupling for various $T_1$ values.
        (i) Simulated CZ gate duration vs qubit–coupler coupling for the same range of $T_1$ times.
    }
    \label{fig:fig_sm_simulations}
\end{figure*}

\subsection{$ZZ$ interaction and transverse coupling}
\label{subsec:appendix:zz_and_g}

For a system comprising two transmon qubits coupled via a coupler, the effective two-qubit Hamiltonian can be written as:
\begin{equation}
    \hat{H}/\hbar = \sum_{i=1,2} \frac{1}{2} \Tilde{\omega}_i \hat{\sigma}_j^z + g_{XY} \left( \hat{\sigma}_1^+ \hat{\sigma}_2^- + \hat{\sigma}_2^-\hat{\sigma}_1^+ \right),
\end{equation}
where $\Tilde{\omega}_i$ is the angular frequency of the eigenmode corresponding to $\mathrm{Q}_i$, $\hat{\sigma}_j^z$ is the Pauli $Z$ operator, $\hat{\sigma}_1^+$ and $\hat{\sigma}_2^- $ are the raising and lowering operators, and $g_{XY}$ is the effective transverse coupling strength between the qubits. In our device, the floating coupler frequency $\omega_\mathrm{c}$ is below the bare qubit frequencies $\omega_i$ for $i \in \{ 1, 2\}$, resulting in a transverse coupling strength that is approximately given by~\cite{marxer2023long, yan2018tunable}:
\begin{equation}
    g_{XY} \approx g_{12} - \frac{g_\mathrm{1c} g_\mathrm{2c}}{2} \left( \frac{1}{\Delta_{1c}} + \frac{1}{\Delta_{2c}} - \frac{1}{\Sigma_{1c}} - \frac{1}{\Sigma_{2c}} \right),
\end{equation}
where $\Delta_{ic} = \omega_i - \omega_\mathrm{c}$ and $\Sigma_{ic} = \omega_i + \omega_\mathrm{c}$.

The $ZZ$ interaction is defined as:
\begin{equation}
    \zeta = \left( \omega_{\ket{11}} - \omega_{\ket{01}} \right) - \left( \omega_{\ket{10}} - \omega_{\ket{00}} \right), 
\end{equation}
where $\omega_{\ket{ij}}$ represents the angular frequency of the corresponding eigenstate. An analytical expression for $\zeta$ can be obtained using a Schrieffer–Wolff transformation that retains terms up to at least fourth order in the coupling strength~\cite{sung2021realization, chu2021}. This analysis reveals that, for most qubit–qubit detunings, the points where $\zeta$ is canceled and where the transverse coupling $g_{XY}$ vanishes occur at different coupler frequencies.

\subsection{Assumptions in the Simulation}

To simulate single- and two-qubit gate errors, we use a custom simulation framework built on the open-source Python library QuTiP~\cite{lambert2024qutip}. As described in the main text, our simulations rely on several simplifying assumptions, which we elaborate on here.

First, both qubits are assumed to have identical relaxation and dephasing times: $T_1^{\mathrm{Q1}} = T_1^{\mathrm{Q2}} = T_\phi^{\mathrm{Q1}} = T_\phi^{\mathrm{Q2}}$. While this is not exact, it closely reflects the behaviour observed in our device (see Fig.~\ref{fig:fig_sm_9}). Second, the coupler is assigned a relaxation and dephasing time that is four times shorter than that of the qubits. The reduced relaxation time of the coupler is motivated by its larger energy participation ratio (EPR) in lossy interfaces -- which is approximately four times higher compared to the EPR of the qubits (\SI{0.8e-3}{} vs. \SI{0.2e-3}{}) -- and by coherence measurements of our fabricated device (see Table~\ref{tab:parameters}). Third, all elements -- the two qubits and the coupler -- are modeled as three-level Duffing oscillators. Adding a fourth level did not significantly alter the simulation results. The anharmonicities of the qubits and coupler are set to values matching those of the fabricated device (see Table~\ref{tab:parameters}). Finally, to isolate the impact of the qubit–coupler coupling strength $g_\mathrm{qc}$, we vary the direct qubit–qubit coupling $g_{12}$ in such a way that the coupler idling frequency remains constant across the sweep. This leads to $g_{12}$ values in the range $[0.7, 17.1]~\si{\mega\hertz}$.

\subsection{Single-Qubit Gate Error Mechanisms}

Single-qubit gate errors are primarily influenced by incoherent processes, hybridisation-induced crosstalk, and, for short pulses, leakage to higher energy levels. Hybridisation crosstalk arises from the fact that the qubits are slightly hybridised with each other due to the choice of the idling configuration. Although the coupler is biased to cancel the $ZZ$ interaction ($\zeta = 0$), the eigenstates of the coupled system -- $\ket{00}$, $\ket{01}$, $\ket{10}$, $\ket{11}$ -- are not perfectly aligned with the bare states $\ket{\mathrm{gg}}, \ket{\mathrm{ge}}, \ket{\mathrm{eg}}, \ket{\mathrm{ee}}$~\cite{sung2021realization}. As a result, driving one qubit can unintentionally excite the hybridised mode of the other, even in the absence of classical crosstalk. 

This hybridisation crosstalk depends on (i) the qubit–coupler coupling strength $g_\mathrm{qc}$ (larger $g_\mathrm{qc}$ leads to stronger hybridisation), (ii) the qubit–qubit detuning $\Delta\omega$, and (iii) the gate duration (shorter pulses have broader frequency content and are more likely to excite neighboring eigenstates). To capture these effects, we perform a four-dimensional sweep over $\sqrt{X}$ gate duration, $g_\mathrm{qc}$, relaxation time $T_1$ (which sets both $T_1$ and $T_\phi$), and qubit–qubit detuning $\Delta\omega$, simulating both qubits simultaneously. The drive waveform is a cosine-shaped envelope, matching our experimental implementation.

Figures~\ref{fig:fig_sm_simulations}(a) and (d) show the simulated $\sqrt{X}$ gate error versus gate duration for various $g_\mathrm{qc}$ values, for $\mathrm{Q}_1$ and $\mathrm{Q}_2$, respectively. For small $g_\mathrm{qc}$, the gate error is coherence-limited for durations longer than approximately \SI{16}{\nano\second}. For shorter gates, leakage to the $\ket{2}$ state becomes the dominant source of error. While such leakage could in principle be mitigated using advanced pulse shaping techniques such as FAST DRAG or higher-order DRAG~\cite{hyyppa2024reducing}, in practice we observe no benefit due to mixer compression at high drive powers.

For larger $g_\mathrm{qc}$, hybridisation crosstalk becomes significant. In this regime, the gate error decreases with increasing duration, as longer gates have narrower frequency spectra and interact more adiabatically with unwanted transitions. Oscillations observed in the $\mathrm{Q}_2$ traces originate from transient population of the $\ket{2}$ state due to the proximity of the $\ket{11}$ and $\ket{02}$ levels. These oscillations are absent in $\mathrm{Q}_1$, which is the lower-frequency qubit and does not couple strongly to $\ket{20}$.

From each duration sweep, we extract the gate duration that minimizes the error and plot the corresponding minimum $\sqrt{X}$ gate error as a function of $g_\mathrm{qc}$ for several $T_1$ values [Figs.~\ref{fig:fig_sm_simulations}(b) and (e)]. The overall trend shows increasing gate error with increasing $g_\mathrm{qc}$. At high $g_\mathrm{qc}$, the system is hybridisation-limited, whereas at low $g_\mathrm{qc}$, it becomes coherence-limited -- yielding two distinct slopes.

To evaluate how the $\sqrt{X}$ gate error depends on qubit–qubit detuning, we simulate five values of $\Delta\omega / (2\pi) \in \{\SI{60}, \SI{85}, \SI{110}, \SI{135}, \SI{160}\}~\si{\mega\hertz}$, shown in Figs.~\ref{fig:fig_sm_simulations}(c) and (f). These detunings are symmetrically spaced around $\Delta\omega / (2\pi) = \SI{110}{\mega\hertz}$, where the $\sqrt{X}$ gate error reaches its minimum. This configuration corresponds to both the $\ket{01}$ and $\ket{10}$, as well as the $\ket{11}$ and $\ket{02}$ states being far detuned from each other.

The simulated detuning range of \SI{100}{\mega\hertz} is chosen to reflect realistic constraints in large-scale processors, where nearest-neighbor (NN) and next-nearest-neighbor (NNN) detunings must be engineered to avoid spectral crowding.

The trend of increasing gate error with $g_\mathrm{qc}$ is consistent across all $\Delta\omega$ values. As expected, hybridisation error is minimized near $\Delta\omega / (2\pi) = \SI{110}{\mega\hertz}$ and increases for both lower and higher detunings. Interestingly, we observe that the gate errors at smaller detunings (\SI{60}{\mega\hertz} and \SI{85}{\mega\hertz}) are higher than their larger-detuning counterparts (\SI{135}{\mega\hertz} and \SI{160}{\mega\hertz}), despite their symmetric placement around the optimal point. This asymmetry arises because the effective transverse coupling $g_{XY}$ -- which governs the strength of hybridisation between the qubits -- is stronger at smaller detunings and weaker at larger ones, leading to enhanced crosstalk for low $\Delta\omega$ in our design.

\subsection{Two-Qubit Gate Error Mechanisms}

To estimate the CZ gate errors, $\epsilon_\mathrm{CZ}$, we perform a three-dimensional sweep over gate duration (i.e., both qubit and coupler flux pulse widths), qubit–coupler coupling $g_\mathrm{qc}$, and $T_1$. The qubit flux pulse includes \SI{1}{\nano\second} cosine-shaped rising and falling edges. Slices of the results are shown in Figs.~\ref{fig:fig_sm_simulations}(g–i).

Fig.~\ref{fig:fig_sm_simulations}(g) shows $\epsilon_\mathrm{CZ}$ versus gate duration for various $g_\mathrm{qc}$ at fixed $T_1 = \SI{100}{\micro\second}$. At short durations and low $g_\mathrm{qc}$, coherent under-rotation in the $\mathrm{span}( \{\ket{11}, \ket{02} \})$ subspace dominates. At longer durations or higher couplings, incoherent error becomes the limiting factor. For $g_\mathrm{qc} \gtrsim \SI{60}{\mega\hertz}$, oscillatory behaviour appears due to SWAP-type coherent errors, see definition in Appendix~\ref{subsec:appendix:swap_errors}. These can be minimized through careful gate timing, although leakage to $\ket{20}$ persists and ultimately limits performance.

From each duration sweep, we extract the minimum $\epsilon_\mathrm{CZ}$ per $g_\mathrm{qc}$ and plot it versus coupling strength for several $T_1$ values in Fig.~\ref{fig:fig_sm_simulations}(h). For low $g_\mathrm{qc}$, gate errors drop rapidly due to decreasing gate duration and reduced incoherent error. At higher couplings, the trend saturates as leakage to $\ket{20}$ becomes the dominant contribution. The apparent noise in Fig.~\ref{fig:fig_sm_simulations}(h) is not numerical; rather, it reflects the sensitivity of leakage to subtle changes in the CZ gate duration. This is also evident in Fig.~\ref{fig:fig_sm_simulations}(i), where optimal CZ durations exhibit irregular switching between local minima.

\subsection{Connection to Main Text}

The single-qubit error data shown in Fig.~\ref{fig:fig_1}(b) are generated by fixing $T_1 = \SI{100}{\micro\second}$, determining the optimal $\sqrt{X}$ gate duration per $g_\mathrm{qc}$, averaging over the five $\Delta\omega$ values, and summing errors from both qubits. The CZ gate error data corresponds directly to Fig.~\ref{fig:fig_sm_simulations}(h) at $T_1 = \SI{100}{\micro\second}$.

To compute the Clifford error $\epsilon_\mathrm{Clifford}$, we combine the two contributions using a weighted sum based on Clifford decomposition: 4.65 $\sqrt{X}$ gates and 1.5 CZ gates per Clifford, following the prescription in Ref.~\cite{barends2014superconducting}, with additionally replacing $X$ and $Y$ rotations by $\pi/2$ pulses and virtual-$Z$ gates~\cite{mckay2017efficient}.

For Fig.~\ref{fig:fig_1}(c), we apply a bivariate spline interpolation to the $\epsilon_\mathrm{Clifford}(g_\mathrm{qc}, T_1)$ data, extract the minimum along the $g_\mathrm{qc}$ axis for each $T_1$, and plot the result as the solid black line indicating the optimal coupling strength.

\section{DYNAMICS OF THE REPEATED DIABATIC CZ GATE}
\label{sec:appendix:popex_dynamics}
In this section, we elaborate on the theory behind the standard error amplification experiment [Fig.~\ref{fig:fig_2}(a) and (b)], as well as its extensions for characterizing leakage to the $\ket{02}$ state and the coupler. Firstly, the appropriate frame of reference should be chosen.

\subsection{Choice of the frame of reference}
\label{subsec:appendix:lab_frame}
In this and Sections~\ref{sec:appendix:npopex_model} and~\ref{sec:appendix:out_of_subspace_dynamics}, we primarily work in the laboratory reference frame, unless specified otherwise. This choice is motivated by the nature of the diabatic CZ gate, which is implemented by dynamically tuning the frequencies of the coupler and qubit $\mathrm{Q}_1$~\cite{marxer2023long}. As a result, in the laboratory frame, each of the repeated CZ gates in the same gate sequence can be described by the same unitary $\hat{U}_g$ as in Eq.~\eqref{eq:U_g_definition} of the main text. Moving to a standard computational frame, defined by the frequencies of the two qubits, would cause the underlying Hamiltonian $\hat{H}$ to acquire additional time dependence as a function of the difference of the qubit frequencies. Consequently, each consecutive CZ in the same sequence would be described by a different unitary in the $\mathrm{span}( \{\ket{11}, \ket{02} \})$ subspace. Thus, choosing to instead work in the laboratory frame (or any interaction frame associated with the same frequency $\omega$ for both qubits) greatly simplifies the derivations.

This choice has minor implementation consequences which will be described later in Appendix~\ref{subsec:appendix:palea_phase_tracking}.

\subsection{Standard error amplification}
The goal of this section is to analyze the dynamics of the standard error amplification experiment shown in Fig.~\ref{fig:fig_2}(a) and (b) of the main text. In particular, we seek to understand why the signal-to-noise ratio (SNR) of the population oscillations becomes small near the optimal coupler flux pulse amplitude -- rendering this method inefficient for minimizing leakage.

To this end, we calculate the expectation value $V_n$ after applying $n$ repeated CZ gates to the initial state $\ket{11}$. We start from the unitary $\hat{U}_g$ defined in Eq.~\eqref{eq:U_g_definition} in the main text, acting on the $\mathrm{span}( \{\ket{11}, \ket{02} \})$ subspace and decomposed as $\hat{U}_g = \Rop_z(\alpha) \Rop_x(\theta) \Rop_z(\beta)$, where $\theta$ governs the population transfer between $\ket{11}$ and $\ket{02}$. The expectation value for the population remaining in $\ket{11}$ after $n$ applications of $\hat{U}_g$ is then:
\begin{equation}\label{eq.npopex_reg_expectation}
\begin{aligned}
    V_n(\theta, \alpha, \beta) = \left|\braket{11 | \left[\hat{U}_g(\theta, \alpha, \beta) \right]^n | 11}  \right|^2 \\
    = \left|\braket{11 | \left[\Rop_z(\alpha + \beta) \Rop_y(\theta )\right]^n | 11}  \right|^2.
\end{aligned}
\end{equation}
In the second line, we use the identity
\begin{equation}\label{eq:rz_absorption}
\begin{aligned}
    &\left( \Rop_z(\alpha) \Rop_x(\theta ) \Rop_z(\beta)\right)^n = \\ 
    &\Rop_z(\alpha - \frac{\pi}2) \left( \Rop_z(\alpha + \beta)\Rop_y(\theta )\right)^n \Rop_z(\beta + \frac{\pi}2),
\end{aligned}
\end{equation}
and note that the leading and trailing $Z$ rotations do not affect the expectation value in a $Z$ basis measurement. We define $\phi = \alpha + \beta$ as the effective $Z$ rotation applied in each cycle. Since $\alpha$ and $\beta$ are not relevant for the operation of a CZ gate, $\phi$ remains uncalibrated and effectively random, but constant throughout the experiment.

The single cycle operator $\Rop(\phi, \theta) = \Rop_z(\phi) \Rop_y(\theta)$ can be diagonalized as
\begin{equation}
\Rop(\phi, \theta) = S(\phi, \theta) D(\phi, \theta) S^{-1}(\phi, \theta),
\end{equation}
where
\begin{equation}
D(\phi, \theta) =
\begin{pmatrix}
e^{i \mu} & 0 \\
0 & e^{-i \mu}
\end{pmatrix}
\end{equation}
\begin{equation}
S(\phi, \theta) =
\begin{pmatrix}
\sin(\theta/2) & \sin(\theta/2) \\
\cos(\theta/2) - e^{i(\frac{\phi}2 + \mu)} & \cos(\theta/2) - e^{i(\frac{\phi}2 - \mu)}
\end{pmatrix}.
\end{equation}
The matrices $D(\phi, \theta)$ and $S(\phi, \theta)$ are defined in the \{$\ket{11}, \ket{02}\}$ basis, and the angle $\mu$ satisfies
\begin{equation}\label{eq:mu_definition}
\cos \mu = \cos(\phi/2) \cos(\theta/2).
\end{equation}
Substituting into Eq.~\eqref{eq.npopex_reg_expectation}, the expectation value becomes:
{\small
\begin{equation}\label{eq:expansion_ncycles}
\begin{split}
\left|\braket{11 | \Rop^n(\phi, \theta) | 11} \right|^2 &= \left|
\begin{pmatrix}
1 & 0
\end{pmatrix}
\cdot S(\phi, \theta) D^n(\phi, \theta) S^{-1}(\phi, \theta) \cdot
\begin{pmatrix}
1 \\
0
\end{pmatrix}
 \right|^2 \\
&= \frac{\cos^2(\theta/2) - \cos(2 \mu) + \sin^2(\theta/2) \cos(2 n \mu)}{2 \sin^2 \mu},
\end{split}
\end{equation}
}
This expression can be written as $A + B\cos(2n\mu)$, where $A$ and $B$ are independent of $n$. The signal therefore consists of an oscillation with frequency $2\mu$ and amplitude $B$, which determines the SNR of the coherent error amplification experiment.

We can express the amplitude $B$ as:
 \begin{equation}\label{eq:sine_lorentzian}
     2B =  \frac{1}{1 + \left(\frac{\sin{\frac{\phi}2}}{\tan{\frac{\theta}2}} \right)^2}.
 \end{equation}
This is a Lorentzian function centered around $\phi = 0$, with a full width at half maximum (FWHM) of $2 \tan(\theta/2)$. For arbitrary $\phi$ and small $\theta$, the amplitude $B$ becomes small, leading to poor SNR and ineffective error amplification. This explains the low contrast observed near the optimal coupler flux pulse amplitude (where $\theta$ is small) in Fig.~\ref{fig:fig_2}(b) of the main text.

Finally, we note that the signal oscillates at the frequency $2\mu$ per gate, which satisfies $\theta \leq 2\mu \leq \pi$. Eq.~\eqref{eq:mu_definition} shows that $2\mu = \theta$ is satisfied for $\phi = \alpha + \beta = 0$, but in general $\mu$ depends on the uncontrolled angle $\phi$. This means that the frequency of the measured signal is systematically larger than $\theta$, making it difficult to accurately extract the leakage angle based on the oscillation frequency alone, even if the lowered SNR is overcome by increasing the number of repetitions in the experiment.

\subsection{Extensions to the standard error amplification}\label{subsec:appendix:floquet}
\begin{figure}[tb]
    \centering
    \includegraphics[width=1\columnwidth]{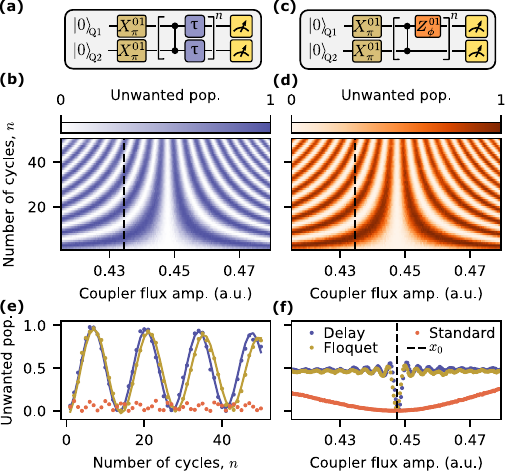}
    \caption{
        (a) Circuit diagram for error amplification with an additional delay $\tau$ inserted after each CZ gate. The system is initialized to $\ket{11}$ using an $X_\pi^{01}$ gate on each qubit. During the delay, the system evolves under the free Hamiltonian.
        (b) Measured unwanted population in the $\ket{02}$ state as a function of the number of cycles $n$, and the coupler flux pulse amplitude.
        (c) Circuit diagram of the Floquet calibration method. A $Z_\phi^{01}$ rotation is interleaved between the CZ gates using two $X_\pi^{01}(\gamma)$ pulses with a relative phase difference $\gamma_0 - \gamma_1 = -\phi/2$ on $\mathrm{Q}_1$.
        (d) Measured unwanted population in the $\ket{02}$ state as a function of CZ gate number and coupler flux amplitude using Floquet calibration. Both (b) and (d) use three-state single-shot readout corrected for assignment error.         
        (e) Vertical slices from (b), (d), and the standard error amplification experiment in Fig.~\ref{fig:fig_2}(b) (main text) at selected flux pulse amplitudes (dashed lines). 
        (f) Unwanted population averaged over the number of cycles, as a function of the coupler flux pulse amplitude. The vertical dashed line indicates the center offset $x_0$ of Lorentzian fits (not shown).
    }
    \label{fig:fig_sm_6}
\end{figure}
For completeness, in this section we present two known extensions to the standard error amplification experiment  that enable accurate measurement of leakage to the $\ket{02}$ state. A key requirement for these methods is that the CZ gate itself remains unmodified during error amplification. Both extensions aim to compensate for a nonzero angle $\phi = \alpha + \beta$, which arises due to the levels $\ket{11}$ and $\ket{02}$ being non-resonant during the final calibrated gate and the additional free evolution during the added buffer time. If $\phi$ is effectively canceled -- such that $\phi \ll \theta$ -- then the signal amplitude in Eq.~\eqref{eq:sine_lorentzian} becomes large, allowing reliable amplification of leakage error.

The first approach, shown in Fig.~\ref{fig:fig_sm_6}(a), introduces a fixed delay $\tau$ after each CZ gate. The delay is chosen to satisfy $\tau ( \omega_{\ket{11}} - \omega_{\ket{02}}) \approx \phi$, where $\omega_{\ket{11}}$ and $\omega_{\ket{02}}$ are the angular frequencies of the corresponding energy levels. This compensates for the accumulated $Z$ rotation in the laboratory frame. Fig.~\ref{fig:fig_sm_6}(b) shows experimental data for a fixed delay of \SI{15.667}{\nano\second}. Near the optimal coupler flux amplitude, the signal amplitude is significantly larger than in the standard error amplification experiment (Fig.~\ref{fig:fig_2}(b) of the main text).

A more robust alternative is Floquet calibration~\cite{arute2020observationseparateddynamicscharge}, where a $Z$ rotation operation with angle $-\phi$ is interleaved between CZ gates to counteract $\phi$, see Fig.~\ref{fig:fig_sm_6}(c). The $Z$ rotation is implemented on $\mathrm{Q}_1$ via two $R_{\gamma}^{01} (\pi) = \Rop_z(\gamma) \Rop_x(\pi) \Rop_z(-\gamma) $ gates with a relative phase difference $\gamma_0 - \gamma_1 = -\phi/2$, effectively realizing an $\Rop_z(-\phi)$ gate. 

We note that in practice the microwave pulses implementing the gates are not zero duration. As a result, there is an additional constant component to the implemented $Z$ rotation angles originating from the free Hamiltonian evolution and possibly other effects, for example a microwave-induced AC Stark shift~\cite{Carroll2022ACStark}. Since the true laboratory frame $\phi$ is unknown, and these additional components only affect the reference, it is always possible to experimentally determine the optimal nominal $Z$ rotation angle. Fig.~\ref{fig:fig_sm_6}(d) shows data using Floquet calibration with an interleaved nominal $Z$ rotation of \SI{-1.122}{\radian}.

Both methods perform significantly better than the standard experiment. As illustrated in Fig.~\ref{fig:fig_sm_6}(e), the oscillation amplitudes at near-optimal coupler flux amplitudes are markedly increased. This results in sharper signals, enabling precise determination of the optimal coupler flux amplitude, as shown in Fig.~\ref{fig:fig_sm_6}(f), where a Lorentzian fit is used to extract the optimal amplitude $x_0$.

However, these extensions are not practical for fine calibration tasks. The protocols require precise tuning of the delay $\tau$ or $Z$ rotation angle $\phi$, because the range over which the SNR is enhanced -- determined by Eq.~\eqref{eq:sine_lorentzian} -- is narrow. Such tuning is time-consuming and sensitive to fluctuations, even if performed adaptively. Moreover, varying gate parameters in an optimization routine would necessitate adaptation of these compensating operations at each point, limiting their usefulness.

\section{PALEA MODEL}
\label{sec:appendix:npopex_model}

In this section, we present the theoretical framework underlying PALEA and demonstrate its specific application to measuring the leakage from $\ket{11}$ to $\ket{02}$. Finally, we compare the accuracy and precision of PALEA to the other error amplification protocols discussed in the previous sections. In Appendix~\ref{subsec:appendix:npopex_echo_out_of_subspace} we show that PALEA is insensitive to other leakage processes, enabling accurate calibration of a single exchange process at a time.

Similarly to the MEADD protocol~\cite{Gross2024}, PALEA can be generally applied to any excitation-preserving entangling gate as long as small off-resonant exchange rate between the two levels can be assumed. The only two required components are the ability to implement the dynamical decoupling step and to average the relevant phase. The protocol is also robust to miscalibration in single qubit rotations used.

\subsection{Defining the components of the PALEA model}
\begin{figure}[tb]
    \centering
    \includegraphics[width=1\columnwidth]{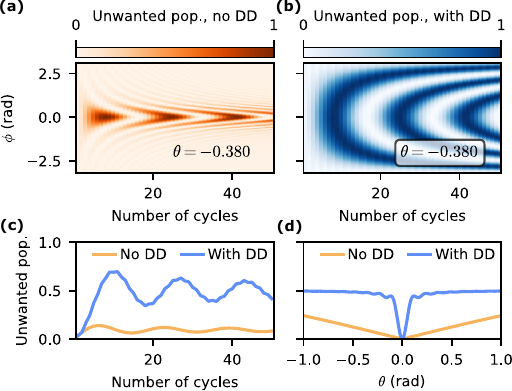}
    \caption{
    (a), (b) Unwanted population as a function of the number of cycles and the $Z$ rotation angle $\phi$, for an over-rotation angle $\theta = -0.380$. Panel (a) shows the case without dynamical decoupling (DD), corresponding to a Floquet-type calibration, while panel (b) shows the case with a DD layer, corresponding to the MEADD-type calibration.
    (c) Unwanted population averaged over $\phi$, plotted as a function of the number of cycles.
    (d) Unwanted population averaged over both $\phi$ and the number of cycles, plotted as a function of the overrotation angle $\theta$.
    }
    \label{fig:npopex_model_comparison}
\end{figure}
Following the introduction of the PALEA protocol in Fig.~\ref{fig:fig_2}(c), we aim to justify the \cref{eq:palea_dd_definition,eq:npopex_expectation,eq:I_n_def_V_n}. First, we define the dynamical decoupling (DD) step $D$ using the single qubit operators describing phased $\pi$ rotations in appropriate subspaces:

\begin{equation}
    D = \Rop_z^{12}(\phi_1)\Rop_x^{12}(\pi)\Rop_z^{12}(-\phi_1) \otimes \Rop_z^{01}(\phi_0)\Rop_x^{01}(\pi)\Rop_z^{01}(-\phi_0),
\end{equation}
where $\Rop_z^{12}$ and $\Rop_x^{12}$ denote rotations around the $Z$ and $X$ axis in the $\mathrm{span}( \{\ket{1}, \ket{2} \})$ subspace, and $\Rop_z^{01}$ and $\Rop_x^{01}$ are the corresponding rotations in the $\mathrm{span}( \{\ket{0}, \ket{1} \})$ subspace.
The dynamical decoupling operator $D$ is designed to keep the state in the subspace $\mathrm{span}( \{\ket{11}, \ket{02} \})$:
\begin{align}
\begin{split}
D\ket{11} =  - e^{i(\phi_1 - \phi_0)}\ket{02}  & \\
D\ket{02}  =   - e^{i(\phi_0 - \phi_1)}\ket{11},
\end{split}
\end{align}
and its action can be written as: 

\begin{equation*}
    D = \Rop_z(\phi_0 - \phi_1) \Rop_x(\pi)  \Rop_z(\phi_1 - \phi_0),
\end{equation*}
recovering Eq.~\eqref{eq:palea_dd_definition}.

Next, for rotations of angle $\pi$ we have $\Rop_z(\lambda) \Rop_x(\pi) =\Rop_x(\pi) \Rop_z(-\lambda)$ for any angle $\lambda$, which allows us to combine the action of $D$ and $\hat{U}_g$ -- defined in Eq.~\eqref{eq:U_g_definition} -- into a single cycle unitary:
\begin{equation}
    \begin{aligned}
    D\hat{U}_g &= \Rop_z(\phi_0 - \phi_1) \Rop_x(\pi)  \Rop_z(\phi_1 - \phi_0) \Rop_z(\alpha)\Rop_x(\theta)\Rop_z(\beta) \\
               &= \Rop_z(2\phi_0 - 2\phi_1 - \alpha - \pi) \Rop_x(\pi - \theta)\Rop_z(\beta + \pi).
    \end{aligned}
\end{equation}
Finally, using Eq.~\eqref{eq:rz_absorption} we recover Eq.~\eqref{eq:npopex_expectation}:

\begin{equation*}
    V_n(\theta, \Delta \phi) = \left|\braket{11 | \left[\Rop_z(\Delta \phi) \Rop_x(\pi - \theta )\right]^n | 11}  \right|^2,
\end{equation*}
where $\Delta\phi = 2 \phi_0 - 2 \phi_1 - \alpha + \beta$.

The importance of the added $\pi$ rotation may be intuitively understood by careful inspection. Eq.~\eqref{eq:npopex_expectation} is a special case of Eq.~\eqref{eq.npopex_reg_expectation}, and the solution proceeds identically. According to Eq.~\eqref{eq:mu_definition}, the cycle eigenvalues satisfy:
\begin{equation}\label{eq:palea_eigenvalue}
    \cos \mu = \cos (\Delta \phi /2) \cos(\pi/2 - \theta/2) = \cos (\Delta \phi /2) \sin(\theta/2).
\end{equation}
Since the measured population depends only on $|2\mu|$, we may, without loss of generality, choose the positive branch $0 \le \mu \le \pi/2$. Taking into account the possible values of $\cos(\Delta\phi/2)$, this leads to the tighter bound $\pi - \theta \le 2\mu \le \pi$, which spans a narrow range if $\theta \ll \pi$.

Thus, Eq.~\eqref{eq:npopex_expectation} describes an iterated two-level rotation in the $\mathrm{span}( \{\ket{11}, \ket{02} \})$ subspace, where the rotation axis is nearly orthogonal to the $Z$ axis when the over-rotation angle $\theta$ is small. 

Following a derivation analogous to Eq.~\eqref{eq:sine_lorentzian}, we obtain the signal amplitude:
 \begin{equation}\label{eq:sine_lorentzian_palea}
     2B = \frac{1}{1 + \left(\frac{\sin{(\Delta \phi/2)}}{\tan (\pi/2 - \theta /2)} \right)^2} \approx 1.
 \end{equation}
This indicates that the iterated rotation yields near-maximal contrast between the populations of $\ket{11}$ and $\ket{02}$ as a function of the number of cycles $n$ for most values of $\Delta\phi$.

Comparing the PALEA protocol to the Floquet calibration discussed in Appendix~\ref{subsec:appendix:floquet}, we find that the cycle unitary in the absence of a dynamical decoupling (DD) layer is dominated by the $Z$ component -- unless the combined phase $\phi = \alpha + \beta \ll \pi$. Consequently, the signal amplitude in such protocols is large only near $\phi \approx 0$ for most values of $\theta$, as illustrated in the simulation in Fig.\ref{fig:npopex_model_comparison}(a) for a fixed over-rotation angle $\theta = -0.380$. In contrast, when a DD layer is applied -- as in the MEADD protocol -- the signal amplitude remains large over a wide range of $\phi$ values and over-rotation angles, as shown in Fig.\ref{fig:npopex_model_comparison}(b).

Averaging over the angle $\phi$ corresponds to the "phase-averaged" aspect of PALEA. Without a DD layer, this averaging suppresses the signal significantly, resulting in a low-contrast response [Fig.~\ref{fig:npopex_model_comparison}(c)]. With a DD layer, however, the same averaging preserves high signal amplitude -- this is the core advantage of the PALEA protocol.

Finally, Fig.~\ref{fig:npopex_model_comparison}(d) shows the signal averaged over both $\phi$ and the number of cycles while sweeping the over-rotation angle $\theta$. In the presence of a DD layer, this produces a sharp and well-defined dip, enabling precise calibration. Without a DD layer, the dip becomes shallow and less informative.

\subsection{Phase tracking in the PALEA protocol}
\label{subsec:appendix:palea_phase_tracking}
The PALEA model is defined in the laboratory frame, for reasons outlined earlier in Appendix~\ref{subsec:appendix:lab_frame}. While the phase difference $\Delta \phi$ may vary between repetitions of the experiment, it must remain constant across the repeated cycles within a single gate sequence. Specifically, if the computational frame of each qubit is defined by the frequency of the microwave drive used to implement the required rotation, then in the laboratory frame the phase difference $\Delta \phi$ evolves at a frequency given by $\frac{1}{\hbar}(E_{\ket{11}} - E_{\ket{02}})$, where $E_{\ket{11}}$ and $E_{\ket{02}}$ are the energies of the respective states. In most experimental setups, the phase of the microwave reference used to implement the single qubit gates is allowed to progress freely at its frequency. Without correction, this causes $\Delta \phi$ to drift between cycles. To prevent this undesired evolution, one must track the accumulated phase and actively compensate for it. In practice, this can be achieved by regressing the microwave phase by $\phi_i$ after each cycle, with $\phi_1 = -\tau \Tilde{\omega}_1$ for qubit $\mathrm{Q}_1$ and $\phi_2 = -\tau (\Tilde{\omega}_2 + \alpha_2)$ for qubit $\mathrm{Q}_2$, where $\tau$ is the duration of one cycle. In most experimental frameworks it only requires setting the virtual Z rotations attached to the CZ gate to the above values.


\subsection{Solving for the expectation value in the PALEA model}\label{subsec:appendix:npopex_echo_solving}
Since the eigenvalues and eigenvectors of the cycle with the added $\pi$ rotation is less sensitive to $\Delta \phi$, averaging over $\Delta \phi$ results in a high contrast signal from which $\theta$ can be extracted. To justify this intuition and solve the integral in Eq.~\eqref{eq:I_n_def_V_n}, we prove the following lemma:

\begin{lemma}\label{lm:integral_proof}
For $\hat{U} \in \mathrm{SU}(2)$, $n \in \mathbb{N}$, the average of the matrix element
\begin{equation}\label{eq:av_definition}
\left|\braket{0 | \left[\Rop_z(\phi) \hat{U} \right]^n | 0}  \right|^2
\end{equation}
over $\phi$ is given by
\begin{equation}
I_n(\theta) = 1 - \sin^2(\theta/2) \sum_{m = 0}^{n - 1} P_m(\cos \theta)
\end{equation}
where $\theta$ is such that $\hat{U} = \Rop_z(\alpha) \Rop_y(\theta)\Rop_z(\beta)$ and $P_m$ is the $m$-th Legendre polynomial.
\end{lemma}
\begin{proof}
The average can be expressed as an integral using the ZYZ Euler angle decomposition for $\hat{U}$:

\begin{equation}\label{eq:I_n_definition}
I_n(\theta) = \frac{1}{2\pi} \int_{0}^{2\pi} \mathrm{d} \phi \left|\braket{0 | \left[\Rop_z(\phi) \Rop_z(\alpha) \Rop_y(\theta)\Rop_z(\beta) \right]^n | 0} \right|^2
\end{equation}
Note 
\begin{equation}\label{eq:angle_composition}
\begin{split}
\left|\braket{0 | \left[\Rop_z(\phi) \Rop_z(\alpha) \Rop_y(\theta)\Rop_z(\beta) \right]^n | 0}  \right|^2 = & \\= \left|\braket{0 | \left[\Rop_z(\phi + \alpha + \beta) \Rop_y(\theta) \right]^n | 0}  \right|^2
\end{split}
\end{equation}
utilizing Eq.~\eqref{eq:rz_absorption}, and we can define:
\begin{equation}
\Rop(\phi, \theta) = \Rop_z(\phi) \Rop_y(\theta).
\end{equation}
\subsubsection{Expanding the integral}

Following the derivation in Eqs.~\eqref{eq.npopex_reg_expectation} to \eqref{eq:expansion_ncycles}, the integrand of Eq.\ \eqref{eq:I_n_definition} can be expanded as
{\small
\begin{equation}
\begin{split}
\left|\braket{0 | \Rop^n(\phi, \theta) | 0} \right|^2 &= \left|
\begin{pmatrix}
1 & 0
\end{pmatrix}
\cdot S(\phi, \theta) D^n(\phi, \theta) S^{-1}(\phi, \theta) \cdot
\begin{pmatrix}
1 \\
0
\end{pmatrix}
 \right|^2 \\
&= \frac{\cos^2(\theta/2) - \cos(2 \mu) + \sin^2(\theta/2) \cos(2 n \mu)}{2 \sin^2 \mu},
\end{split}
\end{equation}
}
leading to
{\small 
\begin{equation}\label{eq:I_n_expanded_integrand}
I_n(\theta) = \frac{1}{2\pi} \int_{0}^{2\pi} \mathrm{d} \phi \frac{\cos^2(\theta/2) - \cos[2 \mu(\phi)] + \sin^2(\theta/2) \cos[2 n \mu(\phi)]}{2 \sin^2[\mu(\phi)]}
\end{equation}
}

\subsubsection{Generating function}
We define the generating function $I(z; \theta)$ as a formal power series
\begin{equation}
I(z; \theta) = \sum_{n = 0}^{\infty} I_n(\theta) z^n.
\end{equation}
Substituting here $I_n(\theta)$ from Eq.\ \eqref{eq:I_n_expanded_integrand}, we obtain
{\small
\begin{equation}
\begin{split}
I(z; \theta) = \frac{1}{4 \pi} \int_{0}^{2\pi} \mathrm{d} \phi \frac{\cos^2(\theta/2) - \cos[2 \mu(\phi)] }{\sin^2[\mu(\phi)]} \sum_{n = 0}^{\infty} z^n +  \\ + \frac{\sin^2(\theta/2)}{4 \pi} \int_{0}^{2\pi} \mathrm{d} \phi \frac{\sum_{n = 0}^{\infty} \cos[2 n \mu(\phi)] z^n}{\sin^2[\mu(\phi)]}.
\end{split}
\end{equation}
}
The first term contains the geometric series
\begin{equation}\label{eq:geometric_series}
\sum_{n = 0}^{\infty} z^n = \frac{1}{1 - z}
\end{equation}
while the second contains a related series
{\small
\begin{equation}
\begin{split}
\sum_{n = 0}^{\infty} \cos(2 n \mu) z^n &= \frac{1}{2} \sum_{n = 0}^{\infty} \left[\left(e^{2 i \mu} z \right)^n + \left(e^{-2 i \mu} z \right)^n \right] \\ &= \frac{1}{2} \left(\frac{1}{1 - e^{2 i \mu} z} + \frac{1}{1 - e^{-2 i \mu} z} \right) \\
&= \frac{1 - \cos(2 \mu) z}{1 - 2 \cos(2 \mu) z + z^2} \\
&= \frac{1 + z - 2 z \cos^2 \mu}{(1 + z)^2 - 4 z \cos^2 \mu}.
\end{split}
\end{equation}
}
Using these, the generating function becomes
{\small
\begin{equation}
\begin{split}
I(z; \theta) &= \frac{1}{4 \pi} \int_{0}^{2\pi} \mathrm{d} \phi \frac{\cos^2(\theta/2) - \cos[2 \mu(\phi)] }{\sin^2[\mu(\phi)]} \frac{1}{1 - z} \\
&+ \frac{\sin^2(\theta/2)}{4 \pi} \int_{0}^{2\pi} \mathrm{d} \phi \frac{1}{1 - \cos^2[\mu(\phi)]} \frac{1 + z - 2 z \cos^2[\mu(\phi)]}{(1 + z)^2 - 4 z \cos^2[\mu(\phi)]}.
\end{split}
\end{equation}
}
Performing partial fraction decomposition to the integrand of the second term by treating $\cos^2[\mu(\phi)]$ as the independent variable, we can further simplify this as
{\small
\begin{equation}
\begin{split}
I(z; \theta) = \frac{1}{1 - z} - \frac{z (1 + z)}{1 - z}\frac{\sin^2(\theta/2)}{2\pi} \\
\times \int_{0}^{2\pi} \frac{\mathrm{d} \phi}{(1 + z)^2 - 4 z \cos^2[\mu(\phi)]}.
\end{split}
\end{equation}
}

Based on Eq.\ \eqref{eq:mu_definition}, we have
{\small
\begin{equation}
\cos^2[\mu(\phi)] = \cos^2(\phi/2) \cos^2(\theta/2) = \frac{\cos^2(\theta/2)}{2} \left[1 + \cos( \phi) \right].
\end{equation}
}

Substituting then leads to
{\small
\begin{equation}
\begin{split}
I(z; \theta) = \frac{1}{1 - z} - \frac{z (1 + z)}{1 - z}\frac{\sin^2(\theta/2)}{2 \pi} \\
\times \int_{0}^{2 \pi} \frac{\mathrm{d} \phi}{(1 + z)^2 - 2 z \cos^2(\theta/2) + 2 z \cos^2(\theta/2) \cos \phi}.
\end{split}
\end{equation}
}

The remaining integral can be converted to a contour integral along the unit circle in the complex plane by defining $w = e^{i \phi}$, giving us
{\small
\begin{equation}
\begin{split}
I(z; \theta) = \frac{1}{1 - z} - \frac{z (1 + z)}{1 - z}\frac{\sin^2(\theta/2)}{2 \pi i} \\
\times \oint_{|w| = 1} \frac{\mathrm{d} w}{-z \cos^2(\theta/2) (w^2 + 1) + [(1 + z)^2 - 2 z \cos^2(\theta/2)] w}.
\end{split}
\end{equation}
}

The integrand has two first-order poles at
\begin{equation}
w_\pm = -1 + \frac{1 + z}{2z} \times \frac{1 + z \pm \sqrt{1 - 2 z \cos \theta + z^2}}{\cos^2(\theta/2)},
\end{equation}
but only $w_-$ is inside the unit circle. Using the residue theorem, we finally obtain
\begin{equation}\label{eq:generating_function_simplified}
I(z; \theta) = \frac{1}{1 - z} - \frac{z}{1 - z}\frac{\sin^2(\theta/2)}{\sqrt{1 - 2 z \cos \theta + z^2}}.
\end{equation}

\subsubsection{Extracting the integrals}
Utilising again the geometric series \eqref{eq:geometric_series} as well as the generating function 
\begin{equation}
\frac{1}{\sqrt{1 - 2 x t + t^2}} = \sum_{n = 0}^{\infty} P_n(x) t^n
\end{equation}
of the Legendre polynomials, we can expand the generating function \eqref{eq:generating_function_simplified} as
\begin{equation}
I(z; \theta) = \sum_{n = 0}^{\infty} z^n - \sin^2(\theta/2) \sum_{n, m = 0}^{\infty} P_m(\cos \theta) z^{n + m + 1}.
\end{equation}
Changing the summation variables from $(n, m)$ to $(n + m, m)$, this can be further written as
\begin{equation}
I(z; \theta) = 1 + \sum_{n = 1}^{\infty} \left[1 - \sin^2(\theta/2) \sum_{m = 0}^{n - 1} P_m(\cos \theta) \right] z^n.
\end{equation}
From here we can immediately read that
\begin{equation}\label{eq:I_0_simplified}
I_0(\theta) = 1
\end{equation}
and
\begin{equation}\label{eq:I_n_simplified}
I_n(\theta) = 1 - \sin^2(\theta/2) \sum_{m = 0}^{n - 1} P_m(\cos \theta)
\end{equation}
\end{proof}
According to the above lemma, the signal from the phase-averaged experiment is given by:

\begin{equation*}
\begin{split}
I_n(\pi -\theta) = 1 - \cos^2(\theta/2) \sum_{m = 0}^{n - 1} P_m(-\cos \theta) \\
= 1 - \cos^2(\theta/2) \sum_{m = 0}^{n - 1} (-1)^m P_m(\cos \theta)
\end{split}
\end{equation*}

using the parity condition of Legendre polynomials, and we finally recover Eq.~\eqref{eq:I_n_npopex_main}.
\subsection{Analyzing the robustness of PALEA}\label{subsec:appendix:npopex_echo_understanding}
\begin{figure}[tb]
    \centering
    \includegraphics[width=1\columnwidth]{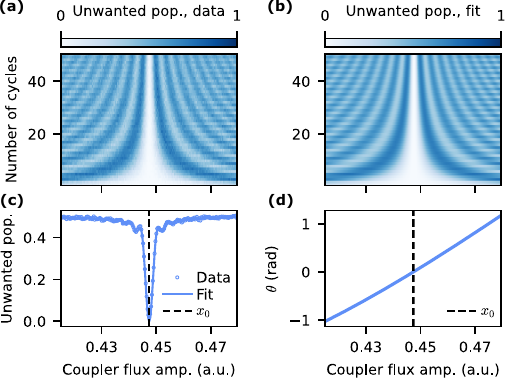}
    \caption{
    (a) Experimental data of the unwanted population as a function of the number of cycles and the coupler flux pulse amplitude. Here, population in the $\ket{02}$ ($\ket{11}$) state is considered unwanted (wanted) for even $n$, and population in $\ket{11}$ ($\ket{02}$) state is unwanted (wanted) for odd $n$.
    (b) Two-dimensional fit of the unwanted population as a function of the number of cycles and the coupler flux pulse amplitude to the model in Eq.~\eqref{eq:I_n_npopex_flipping} .  
    (c) Unwanted population averaged over the number of cycles as a function of the coupler flux pulse amplitude. The solid line indicates the fit to the PALEA model averaged over the number of cycles.
    (d) The over-rotation angles, $\theta$, as a function of the coupler flux pulse amplitude, extracted from experimental data as parameters of the fit shown in (b). The dashed line indicates the optimal coupler flux pulse amplitude, where $\theta=0$.
    }
    \label{fig:npopex_legendre_2d_fit}
\end{figure}
For small $\theta$, Eq.\ \eqref{eq:I_n_npopex_main} describes a signal flipping between 0 and 1. To regularize it, we can consider
\begin{equation}\label{eq:I_n_npopex_flipping}
\Tilde{I}_n(\theta) = (-1)^n\left[I_n(\pi -\theta) - \frac12\right] + \frac12
\end{equation}
instead. It can be shown to have a limiting form when $n \gg 1$ and $\theta \ll 1$:
\begin{equation}\label{eq:I_n_npopex_bessel}
\Tilde{I}_n(\theta) \approx \frac12 - \frac12 J_0(n\theta),
\end{equation}
where $J_0$ is a Bessel function of the first kind. This expression is easier to understand and analyze than the full model. Its maximum is approximately $0.7$, which is another way of confirming that higher SNR can be achieved for rotations which are nearly $\pi$, rather than nearly 0, seen in Fig. \ref{fig:npopex_model_comparison}.

If a measured quantity depends on the phase as a sine, $V \propto \sin(n \theta)$, the phase can only be estimated up to the Nyquist frequency $\frac{\pi}n$. However, the signal following Eq.~\eqref{eq:I_n_npopex_bessel} is not entirely periodic, and only assumes values between 0 and 0.35 once, for values of $n \theta < 1.8$. Thus if a high population of the wanted outcome state is measured after a large number of gates, one might be certain the angle $\theta$ must be small and its absolute value may be calculated from a single phase averaged measurement. This greatly simplifies calibration, as one might fix $n$ at a value guaranteeing required precision, and vary the gate parameters to unequivocally identify a region where $\left|\theta \right|$ is below a certain threshold. This is why a simple sweep of a coupler pulse amplitude, shown in Fig.~\ref{fig:fig_2}(d) works well as a quick calibration experiment.

Instead of fitting a Lorentzian, one might fit the full model [Eq.~\eqref{eq:I_n_npopex_main}] directly to the two-dimensional experimental data. The data seen in Fig. \ref{fig:fig_2} (d) and \ref{fig:npopex_legendre_2d_fit} (a) were collected using single-shot readout calibrated for three-state discrimination and corrected for assignment error. Since the ideal $\pi$ rotation would flip the state between $\ket{11}$ and $\ket{02}$, we define the correct outcome to be $\ket{11}$ for even $n$ and $\ket{02}$ for odd $n$. Then, the unwanted population is defined as:
\begin{equation}\label{eq:error_population}
    UP_n = 1 - \frac{\text{Number of correct outcomes}}{\text{Number of the outcomes 11 or 02}}
\end{equation}
It is an unbiased estimator in presence of incoherent processes causing the decay to the single- and zero-excitation subspaces~\cite{arute2020observationseparateddynamicscharge}. The variance increases as the total number of in-subspace repetitions decreases.  Fig.~\ref{fig:npopex_legendre_2d_fit}(b) shows the fit of the full model [Eq.~\eqref{eq:I_n_npopex_main}] to the data, assuming that $\theta$ depends on the coupler pulse amplitude polynomially. Fig.~\ref{fig:npopex_legendre_2d_fit}(c) shows the data and fit averaged over the number of cycles, and Fig.~\ref{fig:npopex_legendre_2d_fit}(d) the $\theta$ values extracted as fit parameters.

The Legendre series in Eq.~\eqref{eq:I_n_npopex_main} may make it seem this model would be slow to evaluate and therefore slow to fit to the data. However, the evaluation using an appropriate version of the Clenshaw algorithm using the Legendre polynomial recurrence relation~\cite{Clenshaw1955, Ng1968, Saffren1971} has time complexity $\mathcal{O}(n)$ for a single $n$ and thus $\mathcal{O}(n^2)$ for all numbers of cycles until $n$ and is implemented in the \texttt{numpy} Python library as \texttt{polynomial.legendre.legval}. One could also simply use the Bessel approximation from Eq.~\eqref{eq:I_n_npopex_bessel} for fitting.

The standard error in the estimate of $\theta$ is inversely proportional to the derivative:
\begin{equation}\label{eq:I_n_derivative}
    \frac{\partial \Tilde{I}_n}{\partial \theta} \approx -\frac{n}2 J_1(n \theta)
\end{equation}
The Bessel function $J_1(x)$ has an asymptotic form $J_1(x) \sim \left(\frac{x}{2}\right)$ for $x < 1$ \cite{Abramowitz1964Bessel}, so for $n \theta < 1$, the error in the estimate of $\theta$ follows $\mathcal{O}(1/n^2)$, which is the same scaling one could reach with a sine function.

For large $x$, $J_1(x)$ has an envelope given by the absolute value of the Hankel function of the first kind $H_1(x) = J_1(x) + iY_1(x)$, namely $J_1(x) < |H_1(x)|$, and $|H_1(x)|$ has an asymptotic form $|H_1(x)| \sim \sqrt{\frac{2}{\pi x}}$ \cite{Abramowitz1964Bessel}. 

In turn, for large $n$ we have $\frac{\partial \Tilde{I}_n}{\partial \theta} < \sqrt{\frac{n}{2\pi \theta}}$,  recovering the $\mathcal{O}(1/\sqrt{n})$ scaling. Thus, PALEA is at its most useful with regards to precision in characterization of $\theta$ when the product $n\theta$ is small and calibration of the set of pulse parameters for which $\theta = 0$. Fortunately, this is most of its use cases, because for large $\theta$s no error amplification is necessary.

To concretely demonstrate the sensitivity of PALEA, and verify these relationships for the full model  defined in Eq.~\eqref{eq:I_n_npopex_flipping}, further numerical simulations of a realistic calibration procedure may be performed, as demonstrated in the following section.

\subsection{Simulating the estimation power of PALEA}\label{subsec:appendix:npopex_echo_simulating}
\begin{figure}[tb]
    \centering
    \includegraphics[width=1\columnwidth]{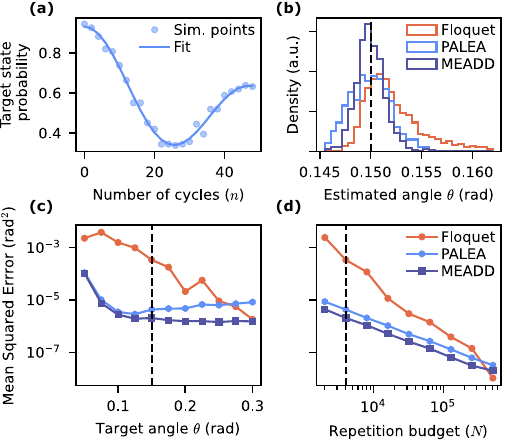}
    \caption{
    Simulations of PALEA, non-adaptive Floquet calibration, and hypothetical MEADD protocol assuming a finite number of repetitions and $0.05$ readout error.
    (a) Simulated frequencies of measuring the target state in PALEA, using total of 4000 repetitions with target angle $\theta=0.15$. The fit to the noisy simulated points is the source of the estimated angle.
    (b) Histograms of estimated $\theta$ values extracted from 1000 runs of Floquet calibration, MEADD and PALEA. MEADD and PALEA both estimate the angle correctly on average, Floquet calibration overestimates it, because the $Z$ angle rotation is not correct due to a finite number of repetitions, introducing a systematic error.
    (c) Mean Squared Error (MSE) produced by each estimator with 4000 total repetitions as a function of the target angle. For small angles, PALEA and MEADD are equally good, but then PALEA signal decreases because it is proportional to a Bessel function $J_0$ and not a cosine. Floquet calibration underperforms because of the systematic error.
    (d) MSE produced by each estimator with thet target angle $\theta=0.15$ as a function of the total repetition budget. MSE of PALEA and MEADD is proportional to $1/N$. The Floquet calibration struggles for small number of repetitions because of the compounding error from incorrect estimation of the $Z$ rotation angle, but it reaches the same sensitivity for large number of repetitions.
    }
    \label{fig:fig_sm_palea_vs_floquet}
\end{figure}
We aim to illustrate the ability of PALEA to measure the angle $\theta$ in the context of a realistic automatic calibration procedure. Such a procedure must be fast, that is, use a limited number of repetitions, each corresponding to one quantum measurement, and not require human intervention. We compare PALEA to the non-adaptive Floquet calibration using $Z$ rotations (Fig.~\ref{fig:fig_sm_6}(c)) and to a MEADD procedure we could perform if we were able to set the relative phases $\phi_0 - \phi_1$ of the dynamical decoupling layer $D$ [Eq.~\eqref{eq:palea_dd_definition}] in the PALEA circuit [Fig.~\ref{fig:fig_2}(c)]. 
We assume a constant repetition budget of $N$ total repetitions for each procedure. For each repetition, we calculate an ideal Born-rule probability $p$ of measuring the target state -- $\ket{11}$ for the procedures described in this work. We then generate a random bit from a Bernoulli distribution with probability $p$, and additionally, with probability $\epsilon$, corresponding to readout error, we flip this bit. The generated bits are then averaged to obtain a data point. If a procedure requires collecting $k$ data points to fit them with a model, the total amount of repetitions $N$ is split equally so that each point receives $N/k$ repetitions.

To perform the Floquet calibration, we need to first find the appropriate angle $\phi$ for the $Z$ rotation. We simulate an experiment using $10$ total cycles, sweeping the angle $\phi$ and fit the generated points with the model given in Eq.~\eqref{eq:expansion_ncycles}. Using the recovered estimate for $\phi$, we then simulate another experiment sweeping the number of cycles from 0 to 49, and fit the data with a sinusoidal model, extracting an estimate for $\theta$. We divide the repetition budget equally between the two experiments.

For MEADD, no additional calibrations are needed. We simulate two experiments using the PALEA circuit, fixing the angle $\phi_0 - \phi_1$ to $0$ for the first, $\pi/2$ for the second, following the iSWAP MEADD procedure from Ref.~\cite{Gross2024}. Because we are not interested in the sign of $\theta$, we skip the third circuit. The number of cycles is swept from $0$ to $48$ with step $2$. The repetition budget is divided equally between the two experiments and each realization of the procedure assumes a different, random value of $\alpha - \beta$ defined in Eq.~\eqref{eq:U_g_definition}. We fit the two data sets with sinusoidal models, extracting two angles $\mu_0$ and $\mu_1$, and calculate $\theta = \sqrt{\mu_0^2 + \mu_1^2}$ following a small angle approximation of Eq.~\eqref{eq:palea_eigenvalue}.

For PALEA, we simulate the signal averaged over $40$ values of $\phi_0 - \phi_1$ and fit the data with the model from Eq.~\eqref{eq:I_n_npopex_main}. The number of cycles is swept identically as in the MEADD simulation. An example of a generated signal and fit is shown in Fig.~\ref{fig:fig_sm_palea_vs_floquet}(a) using $N=4000$ and target $\theta=0.15$. 

We perform each procedure 1000 times, generating estimates of $\theta$ and present their histograms in Fig.~\ref{fig:fig_sm_palea_vs_floquet}(b). PALEA and MEADD both estimate $\theta$ correctly, but the non-adaptive Floquet calibration seems to systematically overestimate it. This is because the fitted angle $\mu$ follows Eq.~\eqref{eq:mu_definition}, which is larger than $\theta$ if $\phi \neq 0$ -- the estimation of $\phi$ is itself erroneous if not enough repetitions are used for it. This demonstrates the fundamental reason why this procedure is too resource-intensive to be used in automatic calibration -- the two estimation errors compound. 

The target angle $\theta$ is then varied, shown in Fig.~\ref{fig:fig_sm_palea_vs_floquet}(c), revealing that PALEA and MEADD have equal sensitivity for small angles, but later the lower signal in PALEA, following a Bessel function, causes lowered SNR compared to a pure sinusoidal oscillation, as discussed following Eq.~\eqref{eq:I_n_derivative}. The Floquet calibration yields an order of magnitude higher Mean Squared Error (MSE), because of the systematic error inherent to the procedure.

Finally, we investigate the MSE when the total repetition budget $N$ is varied. As shown in Fig.~\ref{fig:fig_sm_palea_vs_floquet}(d), the Floquet calibration eventually reaches the same MSE as the two other methods, but only after nearly a million total repetitions.

As discussed in Appendix~\ref{subsec:appendix:npopex_echo_understanding}, PALEA is at its best when estimating smaller angles $\theta$. This simulation does not include any decoherence errors, which would normally limit the total number of cycles performed, so the relevant quantity here is the product of the angle $\theta$ and the maximum number of cycles used. PALEA performs best when this product is approximately 5. In a realistic calibration procedure, a previous coarse experiment using just one CZ gate would have been performed, identifying a region where the angle $\theta$ is small. Thus the increased MSE for larger angles is not relevant for such a procedure. 

The PALEA procedure suffers from slightly lowered SNR in ideal conditions, but in realistic conditions the $\ket{11}$ - $\ket{02}$ dynamics are not the only dynamics present. The CZ gate suffers from out-of-subspace effects which may pollute the calibration, and averaging the phases might be helpful to reduce experimental sensitivity to this pollution.

\section{OUT-OF-SUBSPACE DYNAMICS}\label{sec:appendix:out_of_subspace_dynamics}
In addition to the main exchange dynamics between the $\ket{11}$ and $\ket{02}$ states, CZ gates might induce other unwanted transitions, such as leakage to the coupler or other neighbouring modes, if any are present. First, we will review the tools to measure the effect of those additional dynamics and then discuss how to prevent them from introducing systematic errors in the CZ gate calibration.
\subsection{Measuring coupler leakage}\label{subsec:appendix:measuring_coupler_leakage}
\begin{figure}[tb]
    \centering
    \includegraphics[width=1\columnwidth]{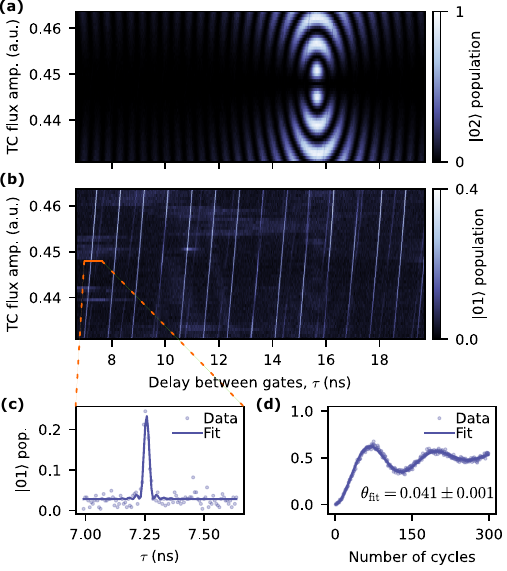}
    \caption{
        (a) Measured population of $\ket{02}$ and (b) $\ket{01}$ versus flux pulse amplitude of the tunable coupler (TC) and delay between CZ gates, $\tau$. The number of cycles is fixed to 30.
        (c) Zoom in on one of the peaks in population of $\ket{01}$ for the optimal CZ gate amplitude, indicated by the orange segment in (b). The solid line shows a fit to the model in Eq.~\eqref{eq:expansion_ncycles} with \SI{30} cycles.
        (d) Evolution of the population of $\ket{01}$ as a function of the number of cycles, using the delay $\tau =$ \SI{7.26}{\nano\second}. The solid line indicates a fit to a model including the coherent coupler leakage to $\ket{011}$ and the incoherent damping involving $\ket{001}$.
    }
    \label{fig:fig_sm_coupler}
\end{figure}
While the extensions of the standard error amplification experiments discussed in Appendix~\ref{sec:appendix:popex_dynamics} are less practical for leakage amplification calibration, they are effective tools for probing out-of-subspace dynamics, including coherent leakage to the coupler~\cite{sung2021realization}.

The goal is to observe coherent population exchange between the $\ket{101}$ and $\ket{011}$ states (with notation $\ket{\mathrm{Q}_1, \mathrm{TC}, \mathrm{Q}_2}$) -- i.e., excitation transfer from $\mathrm{Q}_1$ to the coupler. This two-level system, with interaction strength $g_\mathrm{1c}$, can be described with a single Hamiltonian:
\begin{equation}\label{eq:coupler_hamiltonian_def}
    \hat{H}/\hbar = \begin{pmatrix}
        0 & g_{1\mathrm{c}} \\
        g_{1\mathrm{c}} & \omega_\mathrm{c} (t) - \omega_1
    \end{pmatrix},
\end{equation}
where $\omega_\mathrm{c} (t)$ is the time dependent coupler frequency and $\omega_1$ is the static qubit frequency. Although the coupling $g_{1\mathrm{c}}$ depends on $\omega_\mathrm{c}(t)$ as $g_{1\mathrm{c}}(t) = \beta_{1\mathrm{c}} \sqrt{\omega_1 \omega_\mathrm{c}(t)}$~\cite{marxer2023long}, we approximate it as constant. Transforming to an interaction frame with $\hat{H}_0(t) = (\omega_\mathrm{c}(t) - \omega_1) \ket{011}\bra{011}$ and corresponding unitary
\begin{equation}
    \hat{U}_\mathrm{0} =  e^{-i\int^t \mathrm{d}t'\hat{H}_0/\hbar}  = \exp(- \frac{i}2 \hat{I} \left[ -\omega_1 t + \int^t \mathrm{d}t' \omega_\mathrm{c}(t') \right] ), 
\end{equation}
where $\hat{I}$ is the identity operator, yields a Landau-Zener Hamiltonian:
\begin{equation}
\label{eq:H_landau-zener}
    \hat{H}/\hbar  = \frac12 \left[ (\omega_\mathrm{c} (t) - \omega_1) \hat{\sigma}_z + 2g_{1c} \hat{\sigma}_x\right].
\end{equation}

Analogously to Eq.~\eqref{eq.npopex_reg_expectation}, the unitary $U_\mathrm{c} = e^{-i\int^t \mathrm{d}t'\hat{H}/\hbar}$ can be decomposed into $\Rop_z$ and $\Rop_y$ rotations as:
\begin{equation}\label{eq:coupler_unitary}
    \hat{U}_\mathrm{c}  = \Rop_z (\alpha_\mathrm{c}) \Rop_x (\theta_\mathrm{c}) \Rop_z (\beta_\mathrm{c}),
\end{equation}
where angle $\theta_\mathrm{c}$ quantifies the leakage process between the $\ket{101}$ and $\ket{011}$ states, and the angles $\alpha_\mathrm{c}$ and $\beta_\mathrm{c}$ are the $Z$ axis rotation angles. Again, the goal is to amplify the rotation $\Rop_x (\theta_\mathrm{c})$ for accurate measurement of the coupler leakage.

Although Slepian-shaped flux pulses reduce the leakage to the coupler~\cite{martinis2014fast, sung2021realization}, some leakage will persist. To coherently amplify the angle $\theta_\mathrm{c}$, the total $Z$ rotation $\alpha_\mathrm{c} + \beta_\mathrm{c}$ must vanish, as is evident from Eq.~\eqref{eq:coupler_unitary}. This can be accomplished by inserting delays or phase corrections between the CZ gates. Here, we choose the delay approach, shown in Fig.~\ref{fig:fig_sm_6}(a).

We perform the experiment using 30 cycles of CZ gates interleaved with delays. The measurement of $\ket{02}$ [Fig.~\ref{fig:fig_sm_coupler}(a)] reveals the coherent exchange between the $\ket{11}$ and $\ket{02}$ states described earlier (vertical slices of this figure correspond to horizontal slices of Fig.~\ref{fig:fig_2}(b) and Fig.~\ref{fig:fig_sm_6}(b) ). In contrast, the population of $\ket{01}$, seen in Fig.~\ref{fig:fig_sm_coupler}(b), shows diagonal stripes. We identify these stripes as leakage to the $\ket{011}$ state, i.e., leakage to the coupler. A Fourier transform of this signal along the delay axis reveals a frequency of \SI{1.129}{\giga\hertz}, matching the qubit–coupler detuning (\SI{4.799} - \SI{3.672}{\giga\hertz} = \SI{1.127}{\giga\hertz}; see Table~\ref{tab:parameters}), confirming that the peaks originate from the coupler mode. The precise delays at which the peaks occur depend on the amplitude of the coupler flux pulse approximately linearly. Note that the experiment is effectively a discrete pulsed spectroscopy, revealing a plethora of possible processes, which interfere with the coupler leakage process.

Next, we focus on characterizing the leakage error using the amplitude corresponding to the final calibrated CZ gate, which minimizes the $\ket{02}$ leakage. We select the delay of \SI{7.26}{\nano\second}, at which the coupler leakage process is amplified (Fig.~\ref{fig:fig_sm_coupler}(c)), by modelling the data using Eqs.~\eqref{eq:mu_definition}-\eqref{eq:expansion_ncycles}, and setting $\phi = \tau (\omega_1 - \omega_c)$.

Finally, we measure the population of $\ket{01}$ versus the number of CZ gates, shown in Fig.~\ref{fig:fig_sm_coupler}(d). The measured population consists of the sum of the $\ket{011}$ state population, involved in the coherent process, and $\ket{001}$ state population, where the higher frequency qubit may have incoherently decayed. We model the $\ket{01}$ population assuming the $\ket{101}$ - $\ket{011}$ is the only coherent exchange, characterized by $Y$ and $Z$ rotation angles, accompanied by the incoherent amplitude and phase damping dynamics. We obtain $\theta_\mathrm{c} = 0.041 \pm 0.001$, $\alpha_\mathrm{c} + \beta_\mathrm{c} = 0.025 \pm 0.001$. The signal in Fig.~\ref{fig:fig_sm_coupler}(d) decays mostly due to dephasing in the $\mathrm{span}( \{\ket{101}, \ket{011} \})$ subspace. Additionally, due to limitation in the temporal resolution of the used AWG, the $Z$ rotation component is not fully cancelled, resulting in a slightly reduced population oscillation amplitude.

We note that since only the states of $\mathrm{Q}_1$ and $\mathrm{Q}_2$ are measured, distinguishing coupler leakage from leakage to other modes becomes more difficult in larger QPUs, where unmeasured spectator modes may interfere with calibration \cite{miao2023overcoming}. In our case, the device is a small test chip optimized for high-fidelity operation, with the coupler as the only relevant leakage mode.

Unfortunately, methods requiring the use of $X$ rotations -- such as PALEA -- could not be used to investigate coupler leakage, as the device is missing a dedicated driveline for the coupler. However, Floquet calibration methods can reveal and quantify leakage to modes outside the computational subspace -- such as to the coupler -- with enough measurement time available. They are however less effective when the goal is to quickly and efficiently isolate and quantify leakage to the $\ket{02}$ state as part of the automated calibration procedure. In such cases, leakage to other spectator modes can interfere with the measurement, potentially distorting the inferred $\ket{02}$ leakage and leading to misleading calibration outcomes. Hence, there is a need for a more robust method that selectively amplifies the target leakage process while remaining insensitive to unwanted dynamics.

\subsection{Coupler leakage effects in PALEA}\label{subsec:appendix:npopex_echo_out_of_subspace}
\begin{figure}[tb]
    \centering
    \includegraphics[width=1\columnwidth]{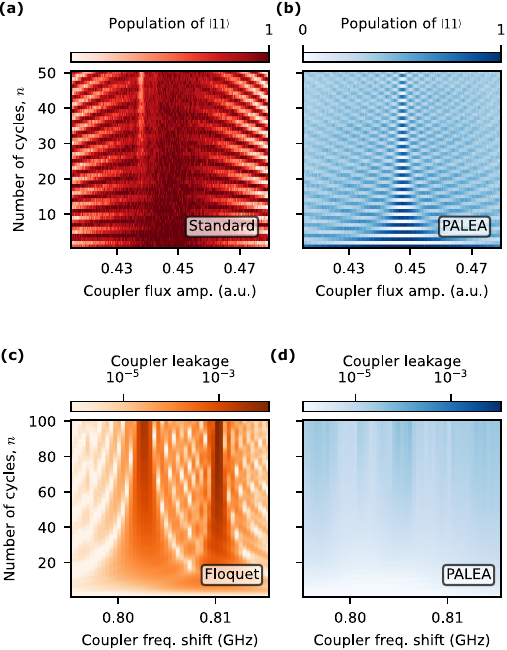}
    \caption{
    (a) Measured population of the state $\ket{11}$ in the standard error amplification experiment, same measurement as in Fig.~\ref{fig:fig_2}(b). Visible vertical line at the amplitude 0.437 is a signature of coupler leakage. (b) Measured population of the state $\ket{11}$ in PALEA, same measurement as in Fig.~\ref{fig:fig_2}(d). No coupler leakage signatures visible.
    (c), (d) Simulated total population of the coupler leakage states $\ket{011}$ and $\ket{112}$ as a function of number of cycles and the coupler frequency shift, using (c) the Floquet calibration sequence and (d) PALEA. 
    Simulations include the lowest four states of both qubits and the coupler, as well as realistic device parameters. The optimal Slepian pulse CZ gate uses coupler frequency shift \SI{0.807}{\giga\hertz}
    }
    \label{fig:npopex_simulations}
\end{figure}
The out-of-subspace dynamics, such as coupler leakage, may bleed into the measurement of the $\ket{11}$ - $\ket{02}$ dynamics, interfering with the data. The goal of this section is to investigate the effect of these dynamics on the results of the PALEA protocol and to confirm that it solves this problem.

Fig.~\ref{fig:npopex_simulations}(a) shows the $\ket{11}$ state population from the same run as Fig.~\ref{fig:fig_2}(b). The coupler leakage signatures also shown in Fig.~\ref{fig:fig_sm_coupler}(b), bleed into this measurement. In contrast, no such traces appear in the $\ket{11}$ state population when measured with PALEA. An intuitive understanding of this may rely on the fact that the amplification of the coupler population occurs only for certain repeated $Z$ rotations in the $\mathrm{span}( \{\ket{101}, \ket{011} \})$ subspace, and PALEA averages over all of them, converting the coupler leakage to incoherent background. This is not entirely trivial however, as the dynamical decoupling step $D$ would excite any potential coupler leakage state $\ket{011}$ to a 4-excitation state $\ket{112}$. Assuming the energy levels $\ket{101}, \ket{002}, \ket{011}$ are all approximately on resonance during the gate, $\ket{112}$ would be the highest energy 4-excitation state. During idling however, this state is lower energy than all 4-excitation states involving the coupler in the ground state. Therefore all these states are crossed by $\ket{112}$ during the coupler flux pulse. The exact dynamical analysis is thus too complicated to be tackled analytically. We observe no pulse-amplitude-dependent coupler leakage in experiments, so we perform numerical simulations. We compare the PALEA dynamics to those of phase-averaged Floquet calibration circuits discussed in Appendix~\ref{subsec:appendix:measuring_coupler_leakage}. Fig.~\ref{fig:npopex_simulations}(c) shows the total population of the states $\ket{011}$ and $\ket{112}$ in a Floquet calibration experiment, where the $Z$ rotation angle has been chosen to show the worst-case scenario leakage. Simulation of the same gates, but using the PALEA circuit, is shown in Fig.~\ref{fig:npopex_simulations}(d). PALEA converts the coherent coupler leakage traces to a two orders of magnitude lower incoherent decay with the number of cycles.

To conclude the analysis of PALEA, we note that by converting various interfering processes into noise, PALEA can be considered robust to coherent interference. Thus by introducing averaging to MEADD, we trade a slightly worse SNR for additional robustness to interfering processes.

\section{OTHER COHERENT ERRORS IN CZ GATE}
\label{sec:appendix:other_coherent}

\subsection{SWAP error}
\label{subsec:appendix:swap_errors}
A CZ gate is an example of a general excitation preserving gate. In addition to two-excitation dynamics discussed above, single-excitation dynamics might additionally contribute to the overall coherent error. These dynamics may contribute to the $ZZ$ rotation angle and the single qubit $Z$ rotations, but those are calibrated separately. For this reason we assume the SWAP error mechanism is a unitary rotation around the $XX + YY$ operator:
\begin{equation}
    \mathcal{E}_s = 
    \begin{bmatrix}
        1 & 0 & 0 & 0 \\
        0 & \cos(\theta_s/2) & -i \sin(\theta_s/2) & 0\\
        0 &  -i \sin(\theta_s/2) & \cos(\theta_s/2) & 0\\
        0 & 0 & 0 & 1
    \end{bmatrix}
\end{equation}
where we define SWAP angle to be $\theta_s$. The simplest way to calculate the gate fidelity contribution is to calculate the entanglement fidelity of the above operator with identity~\cite{Nielsen_2002}, noting that it commutes with the ideal CZ gate. Defining $\ket{\phi}$ as a Bell state in the two copies of the system:
\begin{equation}
    \left| \phi \right\rangle = \frac12 (\left| 00\right\rangle + \left| 11 \right\rangle + \left| 22 \right\rangle + \left| 33 \right\rangle)
\end{equation}
we have
\begin{equation}
\begin{aligned}
F_e (\mathcal{E}_s)=\left|\left\langle \phi \right|I \otimes \mathcal{E}_s\left| \phi \right\rangle\right|^2 &= \left|\frac14 (2 + 2\cos \left(\frac{\theta_s}2 \right))\right|^2 \\
&= \cos^4 \left(\frac{\theta_s}4 \right)
\end{aligned}
\end{equation}
Thus the SWAP contribution to average gate infidelity is
\begin{equation}\label{eq:swap_error_contribution}
\begin{aligned}
    1 - F (\mathcal{E}_s) &= 1 - \frac{4 F_e (\mathcal{E}_s) + 1}{4+1} \\
    &= \frac45 - \frac45 \cos^4 \left(\frac{\theta_s}4 \right)
\end{aligned}
\end{equation}
We measure the SWAP angle using PALEA, preparing the $\ket{10}$ state and using two phased-$X_{\pi}^{01}$ gates as the dynamical decoupling step. Thus the experimental circuit is now identical to one of the MEADD~\cite{Gross2024} circuits, and only differs by the applied phase averaging. The SWAP angle is measured after setting the width of the coupler pulse and calibrating the coupler and qubit pulses for the $ZZ$ rotation angles and the $\ket{11}$ - $\ket{02}$ dynamics. The data analysis is identical to PALEA analysis, replacing the counts of 11 and 02 with counts of 10 and 01. 

We find the SWAP angle below \SI{0.1}{\radian} for all measured widths, and below \SI{0.04}{\radian} for widths longer than \SI{21}{\nano\second}, which does not limit the gate fidelity, see Fig.~\ref{fig:fig_sm_2}(c). We therefore validate the design choices described in Appendix~\ref{sec:appendix:qubit-coupler_coupling}.

\subsection{Contribution to average gate fidelity from coupler leakage}\label{subsec:appendix:coupler_fidelity_contribution}
In Appendix~\ref{subsec:appendix:measuring_coupler_leakage} we demonstrated a measurement of a coupler leakage angle  $\theta_{\mathrm{c}} = 0.041 \pm 0.001$, but we did not calculate the expected contribution to fidelity from this process. The exchange between the basis states $\ket{101}$ and $\ket{011}$ takes the system state from a computational to a non-computational state, making the entire process non-trace preserving. However, if the coupler is traced out, it could be understood as an incoherent process in which the state $\ket{11}$ simply decays to $\ket{01}$. Making that approximation, the average gate infidelity contribution may be calculated by assuming the coupler leakage is an error channel $\mathcal{E}_{\mathrm{c}}$ such that $\mathcal{E}_{\mathrm{c}} (\ket{11}) = \cos ( \theta_\mathrm{c} /2) \ket{11} + \sin ( \theta_\mathrm{c} /2) \ket{01}$, and $\mathcal{E}_{\mathrm{c}} (\ket{10}) = \cos ( \theta_\mathrm{c} /2) \ket{10} + \sin ( \theta_\mathrm{c} /2) \ket{00}$, but it does not affect the other two states. Following the derivation above, it results in the formula identical to Eq.~\eqref{eq:swap_error_contribution}, which for the reported value of $ \theta_{\mathrm{c}}$ yields $\frac45 - \frac45 \cos^4 \left(\theta_{\mathrm{c}}/4 \right) = \SI{1.7(0.1)e-4}{}$. 

This value could be understood as part of the difference between the reported error per ID gate of $\epsilon_\mathrm{ID} = \SI{4.1(0.2)e-4}{}$ and the incoherent contribution to CZ error $\epsilon_\mathrm{CZ}^\mathrm{incoh} = \SI{6.2e-4}{}$ reported in Section~\ref{subsec:cz_error_decomposition}, and represents another way of looking at the same process. Either the qubits hybridise with the coupler, and the excitations in that basis decay with increased rates, or the qubits leak to the coupler's excited state and then the coupler decays.

\subsection{Leakage to $\ket{20}$ state}
\label{subsec:appendix:leakage_20}
For strong couplings and short gate durations, leakage to $\ket{20}$, where the lower qubit's second excited state is populated, may limit the gate fidelity according to simulations. Note that this state's detuning from $\ket{11}$ is approximately double that of the detuning between $\ket{01}$ and $\ket{10}$, while the effective coupling increases by a factor of $\sqrt{2}$. 

We could use either PALEA or the Floquet calibration with delays between the CZ gates, analogous to the one used to measure the coupler leakage shown in Fig.~\ref{fig:fig_sm_coupler}, to measure the rate of exchange between $\ket{11}$ and $\ket{20}$. Note however the assumptions of the model are only satisfied if the stronger exchange in the $\mathrm{span}( \{\ket{11}, \ket{02} \})$ subspace can be eliminated; otherwise the experiments need to be adjusted to take into account three-level dynamics. One such modification could be to utilize a qutrit dynamical decoupling step~\cite{Tripathi2025Qutrit} on both transmons, with a cyclic $X_3 = X_{\pi}^{01} X_{\pi}^{12}$ gate on one and $X_3^{\dagger}$ on the other, cycling through states $\ket{11}$, $\ket{02}$ and $\ket{20}$ --- all of this is however beyond the scope of this work.

We find, using both PALEA and the delay-based amplification, that for gates shorter than \SI{21}{\nano\second}, the coherent leakage to $\ket{20}$ can be measured to be below \SI{0.03}{\radian} per gate, not limiting the fidelity, while for longer gates it is too small to be measurable. At those durations, the leakage to $\ket{02}$ cannot however be entirely eliminated, and is larger than the leakage to $\ket{20}$, violating the model. We conclude that a more appropriate method is to measure the total average leakage rate using the interleaved Leakage Randomized Benchmarking~\cite{wood2018quantification, Rol2019Leakage}, see Section~\ref{subsec:cz_error_decomposition} of the main text. We note that the leakage to the lower qubit's $\ket{2}$ in an RB circuit may originate from the direct leakage to $\ket{20}$, but it can also come from more complicated processes such as leakage to $\ket{02}$, excitation to $\ket{12}$ and population transfer to $\ket{21}$.

\section{CZ GATE OPTIMIZATION}
\label{sec:appendix:cz_gate_optimization}
\begin{figure}[tb]
    \centering
    \includegraphics[width=1\columnwidth]{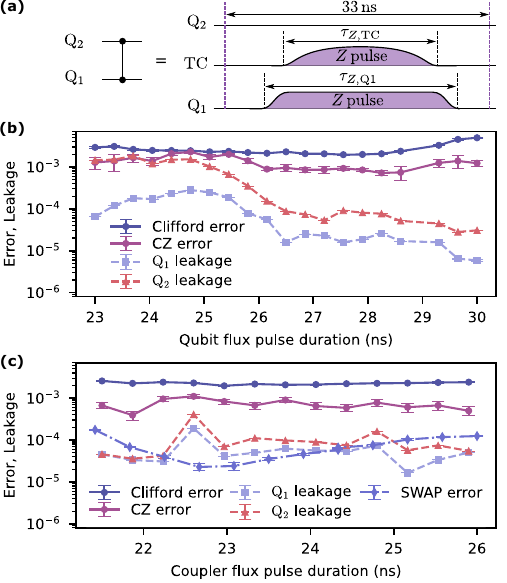}
    \caption{
    (a) Circuit and pulse diagram of the CZ gate. A Slepian-shaped flux pulse with duration $\tau_{Z,\,\mathrm{TC}}$ is applied to the tunable coupler (TC), while a flat-top cosine pulse with a rising edge time of \SI{1}{\nano\second} and a duration $\tau_{Z,\,\mathrm{Q1}}$ is applied to $\mathrm{Q}_1$. $\mathrm{Q}_2$ remains idle during the gate. Virtual $Z$ rotations (not shown) are applied afterward on each qubit to correct for dynamic phases.
    (b) Measured Clifford error, CZ gate error, and leakage to the $\ket{2}$ state for $\mathrm{Q}_1$ and $\mathrm{Q}_2$ as a function of the qubit flux pulse duration $\tau_{Z,\,\mathrm{Q1}}$.  
    (c) Measured Clifford error, CZ gate error, $\ket{2}$ state leakage for both qubits, and SWAP error, plotted as a function of the coupler flux pulse duration $\tau_{Z,\,\mathrm{TC}}$.
    }
    \label{fig:fig_sm_2}
\end{figure}

While calibrating the qubit and coupler flux pulse amplitudes already yields high CZ gate fidelities, further improvement requires careful optimization of the pulse durations. Specifically, both the qubit and coupler flux pulse durations affect leakage to the $\ket{2}$ states and SWAP-type errors, and thus must be tuned to minimize total gate error. The CZ gate circuit and pulse timing are illustrated in Fig.~\ref{fig:fig_sm_2}(a).

\subsection{Optimization of the Qubit Flux Pulse Duration}

We begin by fixing the coupler flux pulse duration to \SI{21}{\nano\second}, a value suggested by simulations to yield high-fidelity gates [see Fig.~\ref{fig:fig_sm_simulations}(g)]. We then sweep the qubit flux pulse duration $\tau_{Z,\,\mathrm{Q1}}$, recalibrate the CZ gate at each step using the PALEA protocol and related calibrations, and finally benchmark the gate using iterative IRB. The results are shown in Fig.~\ref{fig:fig_sm_2}(b).

We extract the $\ket{2}$ state leakage for both qubits using three-state discrimination. For $\mathrm{Q}_1$, this is defined as the total population in states $\ket{20}$, $\ket{21}$, and $\ket{22}$; analogously for $\mathrm{Q}_2$. The leakage per CZ gate is estimated by fitting a linear model to the leakage extracted via iterative IRB, similar to how the CZ fidelity per gate is determined in Section~\ref{sec:cal_opt_and_benchmark_cz_gates} of the main text.

At long qubit pulse durations (\SIrange{28}{30}{\nano\second}), both the Clifford and CZ gate errors rise significantly, despite low leakage. This increase is attributed to coherent errors caused by uncompensated flux pulse distortions in the nanosecond time scales. Since the total gate duration is fixed to \SI{33}{\nano\second}, longer qubit pulses reduce the buffer time between consecutive gates, allowing pulse distortions to bleed into subsequent operations. These distortions manifest as a large quadratic term in the iterative IRB fit.

For short qubit durations (\SIrange{23}{25}{\nano\second}), the dominant error source becomes leakage to the $\ket{2}$ state of $\mathrm{Q}_2$. This behaviour also possibly stems from nanosecond-scale flux distortions, which affect the optimal shape of the coupler pulse. As the Slepian-shaped coupler pulse assumes a static qubit frequency, distortions in the qubit pulse introduce non-ideal Slepian pulse parameters.

The optimal trade-off between leakage and coherent errors is found at a qubit flux pulse duration of approximately \SI{27}{\nano\second}, where both leakage and overall CZ gate error are minimized. At this point, $\ket{2}$ state leakage in both qubits remains below $10^{-4}$.

\subsection{Optimization of the Coupler Flux Pulse Duration}

In the second sweep, we fix the qubit flux pulse duration to \SI{27}{\nano\second} and vary the coupler flux pulse duration $\tau_{Z,\,\mathrm{TC}}$. The results of this sweep, shown in Fig.~\ref{fig:fig_sm_2}(c), exhibit a mostly flat response in Clifford error, CZ gate error, and $\ket{2}$ state leakage.

Because the CZ gate is largely coherence-limited in this regime, the purpose of this sweep is to avoid regions of enhanced leakage -- possibly caused by interactions with weakly coupled parasitic modes -- which appear at seemingly random durations. One such region is seen near \SI{22.6}{\nano\second}, where leakage to the $\ket{2}$ state increases sharply for both qubits, reducing overall gate fidelity. The lowest observed CZ error occurs at a coupler flux duration of approximately \SI{22}{\nano\second}, which we adopt for the final gate parameters reported in the main text [Section~\ref{sec:cal_opt_and_benchmark_cz_gates}].

SWAP-type errors, as described in Appendix~\ref{subsec:appendix:swap_errors}, show an oscillatory behaviour as a function of the coupler flux pulse duration (Fig.~\ref{fig:fig_sm_2}(c)), similar to the oscillations seen in the simulations in Fig.~\ref{fig:fig_sm_simulations}(g). Since the errors due to the SWAP-type errors are mostly at levels below \SI{1e-4}{}, they are negligible with our setup. For coupler flux pulse durations shorter than \SI{20}{\nano\second} (not shown in the figure), the SWAP-type errors become the limiting error source for the CZ gate.

\section{INCOHERENT ERROR OF THE CZ GATE}
\label{sec:appendix:coherence_limit_cz_gate}
\begin{figure}[tb]
    \centering
    \includegraphics[width=1\columnwidth]{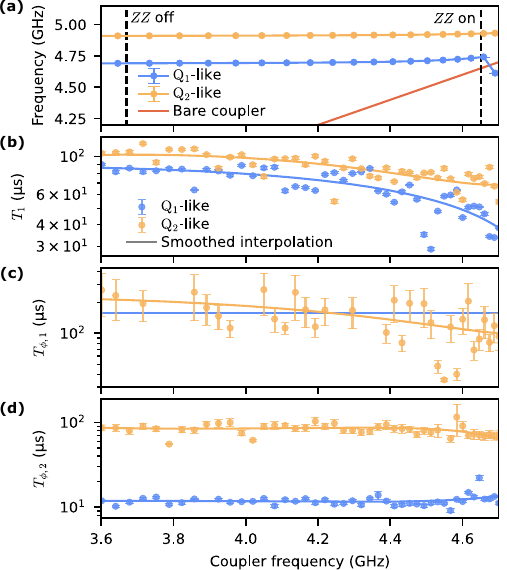}
    \caption{
    (a) Qubit-like mode frequencies and bare coupler frequency as a function of the coupler frequency. Vertical dashed lines indicate the $ZZ$ off and $ZZ$ on configurations.
    (b) Energy relaxation time $T_1$ of both qubits as a function of coupler frequency.  
    (c) Exponential pure dephasing time $T_{\phi,1}$ extracted from Ramsey experiments. Due to $\mathrm{Q}_1$ being dominated by Gaussian dephasing noise, $T_{\phi,1}$ is assumed to be constant at its idling value for $\mathrm{Q}_1$.
    (d) Gaussian pure dephasing time $T_{\phi,2}$ extracted from the same experiment. Dashed lines in (b)–(d) represent smoothed interpolations.
    }
    \label{fig:fig_sm_9}
\end{figure}

To estimate the coherence-limited CZ gate fidelity, we measure the coherence times of the hybridised qubit modes as a function of the coupler frequency and compute an effective average over the flux pulse trajectory, following the approach in Ref.~\cite{marxer2023long}.

Fig.~\ref{fig:fig_sm_9}(a) shows the hybridised qubit frequencies and the bare coupler frequency versus the coupler frequency. The vertical dashed lines indicate the operating points corresponding to the idling configuration ($ZZ$ off) and the CZ gate ($ZZ$ on). During the CZ gate, the coupler is flux-pulsed close to the qubit frequencies to enable strong coupling. To mitigate leakage to the coupler mode, a Slepian-shaped flux pulse is used.

We measure the energy relaxation time $T_1$ by preparing each qubit in its excited state with an $X_\pi$ pulse, then flux pulsing it (along with the coupler) to the same frequencies used during the CZ gate. The qubit flux pulse has the same envelope as during the CZ gate, while the coupler is pulsed with a flat-top cosine envelope, with a rise time of \SI{20}{\nano\second} to ensure adiabatic transitions. This measurement is repeated for each qubit individually, and $T_1$ is extracted via exponential fits to the time-resolved population decay. The results are shown in Fig.~\ref{fig:fig_sm_9}(b).

To extract pure dephasing times, we perform a similar Ramsey experiment with an initial $\sqrt{X}$ pulse. The resulting signal is fitted with the model
\begin{equation}
    p(t) = a(t) \exp\left(-\frac{t}{2 T_1} - \frac{t}{T_{\phi,1}} - \left( \frac{t}{T_{\phi,2}} \right)^2 \right),
\end{equation}
where $T_{\phi,1}$ and $T_{\phi,2}$ are the Gaussian and exponential components of the dephasing time, and $T_1$ is fixed from prior measurements. The exponential and Gaussian dephasing components are shown in Figs.~\ref{fig:fig_sm_9}(c) and (d). 

For $\mathrm{Q}_1$, which is flux-pulsed away from its sweet spot by approximately \SI{80}{\mega\hertz}, the Gaussian component dominates, making $T_{\phi,1}$ extraction unreliable. Therefore, we assume $T_{\phi,1}^\mathrm{Q1}$ to be constant and equal to its value of \SI{158}{\micro\second} at the idling point.

We apply a smoothing interpolation (rather than direct interpolation) to the coherence data to filter out short-term fluctuations, such as those from parasitic two-level systems (TLSs), which are known to cause temporary dips in coherence.

Using the smoothed data, we compute effective coherence times during the CZ gate by integrating the coherence rate along the flux pulse trajectory. The extracted effective values are:
\begin{align*}
    T_{1, \mathrm{eff}}^\mathrm{Q1} &= \SI{63}{\micro\second}, & T_{\phi, 1, \mathrm{eff}}^\mathrm{Q1} &= \SI{157}{\micro\second}, & T_{\phi, 2, \mathrm{eff}}^\mathrm{Q1} &= \SI{12}{\micro\second}, \\
    T_{1, \mathrm{eff}}^\mathrm{Q2} &= \SI{83}{\micro\second}, & T_{\phi, 1, \mathrm{eff}}^\mathrm{Q2} &= \SI{145}{\micro\second}, & T_{\phi, 2, \mathrm{eff}}^\mathrm{Q2} &= \SI{83}{\micro\second}.
\end{align*}

The total incoherent error contribution to the CZ gate is then calculated using the model in Ref.~\cite{abad2025impactofdecoherence}:
\begin{align}
\label{eq:effective_incoherent_error}
    \epsilon_\mathrm{CZ}^\mathrm{incoh} ={}& \frac{3}{10} \frac{\tau}{T_{1, \mathrm{eff}}^\mathrm{Q1}} 
    + \frac{3}{8} \frac{\tau}{T_{\phi, 1, \mathrm{eff}}^\mathrm{Q1}} 
    + \frac{3}{8} \left( \frac{\tau}{T_{\phi, 2, \mathrm{eff}}^\mathrm{Q1}} \right)^2 \nonumber \\
    +& \frac{1}{2} \frac{\tau}{T_{1, \mathrm{eff}}^\mathrm{Q2}} 
    + \frac{31}{40} \frac{\tau}{T_{\phi, 1, \mathrm{eff}}^\mathrm{Q2}} 
    + \frac{31}{40} \left( \frac{\tau}{T_{\phi, 2, \mathrm{eff}}^\mathrm{Q2}} \right)^2,
\end{align}
where $\tau = \SI{33}{\nano\second}$ is the total CZ gate duration.

Substituting the effective coherence times into Eq.~\eqref{eq:effective_incoherent_error} yields an estimated incoherent error of $\epsilon_\mathrm{CZ}^\mathrm{incoh} = \SI{6.2e-4}{}$, in close agreement with the experimentally measured value of $\epsilon_\mathrm{CZ} = \SI{6.5e-4}{}$.

\section{Measuring QNDness}
\label{sec:appendix:qndness}

Quantum non-demolition (QND) measurement refers to a measurement process that preserves the quantum state of the system, meaning the qubit remains in the eigenstate to which it was projected by the readout operation~\cite{lupascu2007quantum}. This property is especially desirable for mid-circuit measurements in quantum error correction, as it removes the need to reset and reinitialize the qubit afterward.

To quantify the QND nature of the readout, we define a “QNDness” metric $\mathcal{Q}$, which reflects the probability that the qubit remains in its original state after measurement, independent of the measurement outcome~\cite{lupascu2007quantum, hazra2025benchmarking}:
\begin{equation}
    \mathcal{Q} = \left( \mathcal{Q}_\mathrm{g} + \mathcal{Q}_\mathrm{e} \right) /2,
\end{equation}
with
\begin{equation}
    \begin{split}
        \mathcal{Q}_\mathrm{g} &= P(\mathrm{g}, 0|\mathrm{g}) + P(\mathrm{g}, 1|\mathrm{g}) \\
        \mathcal{Q}_\mathrm{e} &= P(\mathrm{e}, 0|\mathrm{e}) + P(\mathrm{e}, 1|\mathrm{e}),
    \end{split}
\end{equation}
and $P(a, b|c)$ denotes the probability that the qubit ends up in state $a$ and is assigned outcome $b$, given the qubit is initially in state $c$. For instance, $P(\mathrm{g}, 1|\mathrm{g})$ is the probability of the qubit remaining in $\ket{g}$ but being misassigned to $\ket{1}$.

The aim of this section is to experimentally evaluate the QND character of our readout operations. To this end, we adopt the readout-induced leakage benchmarking (RILB) protocol introduced in Ref.~\cite{hazra2025benchmarking}, with modifications in the analysis, which allow us to extract not only the leakage rate, but also the QNDness $\mathcal{Q}$.

Directly measuring $\mathcal{Q}$ is challenging due to the interplay of various error mechanisms. These include assignment errors such as $P(\mathrm{g}, 1|\mathrm{g})$ and $P(\mathrm{e}, 0|\mathrm{e})$, transition errors such as $P(\mathrm{e}|\mathrm{g})$ and $P(\mathrm{g}|\mathrm{e})$, and leakage to higher excited states such as $P(\mathrm{f}|\mathrm{g})$ and $P(\mathrm{f}|\mathrm{e})$. The RILB protocol is designed to disentangle these effects by analyzing the decay of different correlation types between measurements. These decay rates are sensitive to leakage and transition errors but largely independent of assignment errors, making it possible to extract $\mathcal{Q}$ and the leakage.

\subsection{RILB experiment}
\label{subsec:rilb_experiment}

\begin{figure}[tb]
    \centering
    \includegraphics[width=1\columnwidth]{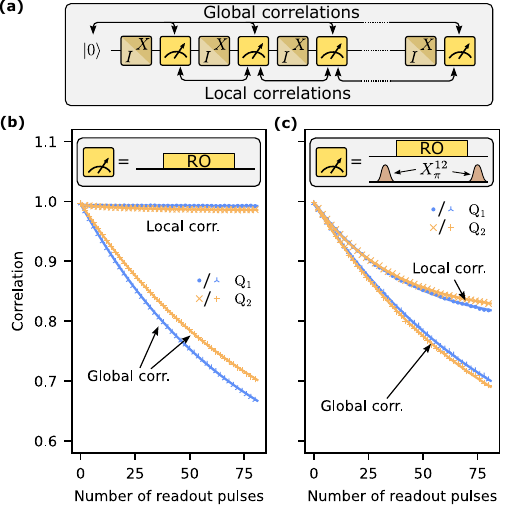}
    \caption{
        (a) Circuit diagram of the readout-induced leakage benchmarking (RILB) experiment. Readout operations are interleaved with randomized bit-flips. Global correlations indicate the correlations between the initial state and expected $n$th readout outcome, whereas local correlations indicate the expected correlations between two consecutive readout outcomes.
        (b), (c) RILB experiment performed on $\mathrm{Q}_1$ and $\mathrm{Q}_2$ with (a) a readout operation consisting of only a readout pulse (RO) optimized for QNDness and (b) a readout operation using the shelving technique via an $X_\pi^{12}$ pulse before and after the readout pulse. The experiments are performed and averaged over 500 randomized bit-flip sequences. The global and local correlations are fitted with exponential decay models. 
    }
    \label{fig:fig_sm_10}
\end{figure}

The circuit diagram for the RILB experiment is shown in Fig.~\ref{fig:fig_sm_10}(a). The protocol consists of $n$ repeated readout operations, interleaved with randomized bit-flips. The outcome of the $n$th measurement is denoted $r_n \in \{0, 1\}$, and the applied pre-measurement operation is $i_n \in \{0, 1\}$, where 1 indicates a bit-flip ($X$ gate) and 0 indicates the identity operation ($I$ gate).

To extract the local correlations between consecutive measurements, we compute the binary "flipped-or-not" quantity:
\[
o_n = r_n \oplus r_{n-1},
\]
and compare it to the applied operation to get the local correlation:
\[
\mathcal{C}_n^\mathrm{local} = 1 - o_n \oplus i_n,
\]
where $\oplus$ denotes the XOR operation. This definition matches that of Ref.~\cite{hazra2025benchmarking} up to a scaling factor.

The original RILB protocol does not allow one to extract the QNDness directly. To solve this issue, 
we extend the protocol to calculate another quantity from the same experimental data. We define a global correlation, which compares the state after $n$ rounds to the initial state $\ket{0}$. Let
\[
 h_n = i_0 \oplus i_1 \oplus \cdots \oplus i_n,
\]
then the global correlation is defined as:
\[
\mathcal{C}_n^\mathrm{global} = 1 - r_n \oplus h_n .
\]

Following the computation of the correlations, we average them with respect to all possible bit-flip sequences, obtaining the expectation values $\left<\mathcal{C}_n^\mathrm{local} \right>$ and $\left<\mathcal{C}_n^\mathrm{global} \right>$. While the correlations are straightforward to compute, modeling their decay as a function of $n$ is non-trivial. In the next sections, we derive analytical models for the expectation values of both global and local correlation decay, allowing us to extract $\mathcal{Q}$, leakage, and seepage.

\subsection{Derivation of the global correlation model}
We use a simple classical hidden Markov chain model to describe the global correlation, assuming the qubit can be in one of $k + 2$ states: $\ket{g}$, $\ket{e}$ or some leakage state $\ket{l_i}$, as in general, more than one leakage state may be populated during the measurement~\cite{hazra2025benchmarking}. We note this assumption of Markovianity may be broken for example if the readout resonator becomes increasingly populated with photons during the experiment. We further assume the single qubit gates used are perfect. We do allow a state preparation error, but assume the qubit is never prepared in a leakage state. 

We start the derivation by defining the bare transition probabilities which do not depend on the assigned outcome:
\begin{equation}
\begin{aligned}
    p_g &= P(\mathrm{e}, 0|\mathrm{g}) + P(\mathrm{e}, 1|\mathrm{g}) \\
    p_e &= P(\mathrm{g}, 0|\mathrm{e}) + P(\mathrm{g}, 1|\mathrm{e}) \\
    L_{g, i} &= P(\mathrm{l_i}, 0|\mathrm{g}) + P(\mathrm{l_i}, 1|\mathrm{g}) \\
    L_{e, i} &= P(\mathrm{l_i}, 0|\mathrm{e}) + P(\mathrm{l_i}, 1|\mathrm{e}) \\
    S_{g, i} &= P(\mathrm{g}, 0|\mathrm{l_i}) + P(\mathrm{g}, 1|\mathrm{l_i}) \\
    S_{e, i} &= P(\mathrm{e}, 0|\mathrm{l_i}) + P(\mathrm{e}, 1|\mathrm{l_i}) 
\end{aligned}
\end{equation}
as well as assignment probabilities, which we assume only depend on the state of the qubit after the readout:
\begin{equation}
\begin{aligned}
    e_0 &= P(0 | \mathrm{e} ) \\
    e_1 &= P(1 | \mathrm{g} ) \\
    w_i &= P(0 | \mathrm{l_i} ) \\
\end{aligned}
\end{equation}
This is a somewhat risky assumption, but in practice we will see $w_i$-s do not enter the final formula. The $e_i$-s can be interpreted as readout SNR errors or overlap errors. 

A hidden Markov chain model~\cite{Rabiner1986AnIT, ghahramani2001hmm} requires an initial state column vector $\mathbf{v}_0$, a transition matrix $T$ and a row emission probability vector $\mathbf{b}$ used to extract the required observation from the model. Thus the expected value of the global (local) correlation will be given by:
\begin{equation}
\begin{aligned}
    \left<\mathcal{C}_n^\mathrm{global} \right> &= \mathbf{b}_g T^n \mathbf{v}_0, \\
    \left<\mathcal{C}_n^\mathrm{local} \right> &= \mathbf{b}_l T^n \mathbf{v}_0, \\
\end{aligned}
\end{equation}
and we need to find $T$ and the two emission probability vectors $\mathbf{b}_g$ and $\mathbf{b}_l$ corresponding to global and local correlations. Assuming Markovianity, the global and local correlation expectations must only depend on the expected state after $n$ steps. The emission probability vectors can be understood as linear functions from the set of states to $\mathbb{R}$, which are used to calculate those expectations. Without specifying the vectors $\mathbf{b}_g$ and $\mathbf{b}_l$, the above equation is simply a statement that those functions are in fact linear, which is a simple consequence of the definition of the expected value.

We define a stochastic population vector $\mathbf{v}$ containing the total probabilities of the system being in a wanted, unwanted or any of the leakage states across all possible realizations of the experiment. The wanted state for a specific realization of the experiment is defined with reference to the number of $X$-gates performed thus far, which is $\ket{g}$ for an even number of $X$-gates preceding the last measurement or $\ket{e}$ for an odd number of $X$-gates. Then, $v^0$ is the total probability that the system is in the wanted computational state and $v^1$ is the total probability that the system is in the unwanted computational state, that is, whatever $v^0$ is not. Finally, $v^{i+2}$ is the probability the system is in the $i$-th leakage state. Thus the length of this vector is $k+2$, where $k$ is the number of leakage states communicating with $\ket{g}$ or $\ket{e}$. 
The initial population vector is
\begin{equation}
    \mathbf{v}_0 = \begin{bmatrix}
        1 - h \\
        h \\
        0 \\
        \vdots
    \end{bmatrix},
\end{equation}
where $h$ is the state preparation error.

Due to a non-standard definition of the basis for the vector $\mathbf{v}$, we need to take additional care to show that the vector of probabilities after $n$ readouts is indeed equal to $T^n \mathbf{v}_0$ for some square stochastic matrix $T$. Formally, we will demonstrate it using the principle of the mathematical induction. 

Clearly, the state before any readouts is the initial state $T^0 \mathbf{v}_0 = \mathbf{v}_0$. We need to show that, assuming the state vector after $n$ readouts is $\mathbf{v}_n = T^n\mathbf{v}_0$, the vector describing the state probabilities after $n+1$ readouts is $\mathbf{v}_{n+1} = T^{n+1}\mathbf{v}_0 = T\mathbf{v}_n$. 

Suppose that after the $n$-th readout, the wanted state is $\ket{g}$. Then, if there is no $X$-gate between the $n$-th and $(n+1)$-st readout, the vector evolves according to the transition matrix:
\begin{equation}
    T_0 = 
    \begin{bmatrix}
        1- p_g - \sum_i^k L_{g, i} & p_e & S_{g, 0} & \cdots\\
        p_g & 1 - p_e - \sum_i^k L_{e, i} & S_{e, 0} & \cdots\\
        L_{g, 0} & L_{e, 0} & \ddots & \\
        \vdots & \vdots & & \ddots
    \end{bmatrix}.
\end{equation}
The unspecified matrix elements are transition probabilities from one leakage state to another, specifically $T_0^{i+2,j+2}$ is the probability of transition from leakage state $i$ to leakage state $j$. They are positive real numbers such that each column sums to $1$.

If there is a flip, the basis of the vector $\mathbf{v}$ is redefined such that its components stay the same after the flip. The basis of $T_0$ is also transformed, which amounts to permuting the first two rows and the first two columns, resulting in the transition matrix:
\begin{equation}
    T_1 = 
    \begin{bmatrix}
        1- p_e - \sum_i^k L_{e, i} & p_g & S_{e, 0} & \cdots\\
        p_e & 1 - p_g - \sum_i^k L_{g, i} & S_{g, 0} & \cdots\\
        L_{e, 0} & L_{g, 0} & \ddots & \\
        \vdots & \vdots & & \ddots
    \end{bmatrix}.
\end{equation}
The transition matrix reflects probabilities of certain transitions occurring, which depends on the presence of the X-gate used to flip the qubit state. Accounting for the $X$-gate being present before the readout with probability of $\frac{1}{2}$, and using the principle of total probability we may write $T^{i,j} = \frac12 T_0^{i,j} + \frac12 T_1^{i,j}$ for each matrix element, leading to the effective transition matrix:
\begin{equation}
    T = \frac12 (T_0 + T_1) = 
    \begin{bmatrix}
        1- p - L & p & S_0 & \cdots \\
        p & 1 - p - L & S_0 & \cdots\\
        L_0 & L_0 & \ddots & \\
        \vdots & \vdots & & \ddots
    \end{bmatrix},
\end{equation}
where $p = \frac12 (p_g + p_e), L = \frac12 (\sum_i^k L_{g, i} + \sum_i^k L_{e, i})$ are average state switch and leakage probabilities, $L_i = \frac12 \left(L_{g,i} + L_{e,i} \right)$ are average leakage rates from computational subspace to leakage state $i$ and $S_i = \frac12 \left(S_{g,i} + S_{e,i} \right)$ are average seepage rates from leakage state $i$ to the computational subspace. The first two columns are identical after the second row, so are the first two rows after the second column. The QNDness defined earlier can be expressed in terms of these quantities: $\mathcal{Q} = 1 - p - L$. The derivation for the case when the wanted state is $\ket{e}$, not $\ket{g}$, follows analogously. Thus, again following the total probability, we can write $\mathbf{v}_{n+1} = T \mathbf{v}_n$, which by mathematical induction implies that $\mathbf{v}_n = T^n \mathbf{v}_0$ for every $n \in \mathbb{N}$.

In addition to a left eigenvector $\mathbf{u}_0 = [1, 1, 1, \cdots]$ with eigenvalue $1$, shared by all stochastic matrices, the matrix $T$ always has a left eigenvector $\mathbf{u}_1 = [1, -1, 0, \cdots]$ with eigenvalue $1 - 2p - L$. This can be verified by direct calculation, and is a consequence of the first two rows being identical after the second column. Additionally, $\mathbf{v}_{-} = \mathbf{u}_1^T$ is a right eigenvector with the same eigenvalue, because the first two columns are identical after the second row. 

Consider now the emission probability vector $\mathbf{b}_g$ corresponding to the global correlation. Assuming the wanted state after $n$ readouts is $\ket{g}$, the first component of the vector $\mathbf{v}_n$ contributes $1-e_1$, second contributes $e_0$, and each $i$-th component afterwards contributes $w_i$. If the wanted state is $\ket{e}$, the contributions change to $1-e_0$, $e_1$ and $1-w_i$. Which state is wanted is random with probability $\frac12$, so the total emission probability vector is 
\begin{equation}
\begin{aligned}
    \mathbf{b}_g &= \left[1 - \frac12(e_0 + e_1), \frac12(e_0 + e_1), \frac12 (w_0 + 1 - w_0), \cdots \right] \\ &= \left[1 - e, e, \frac12, \cdots\right] \\ &=\left(\frac12 - e\right) \mathbf{u}_1 + \left[\frac12, \frac12, \cdots \right] =\left(\frac12 - e\right) \mathbf{u}_1 + \frac12 \mathbf{u}_0,
\end{aligned}
\end{equation}
where $e = \frac12(e_0 + e_1)$. The expected value of the global correlation is then given by
\begin{equation}\label{eq:global_corr_expectation}
\begin{aligned}
    \left<\mathcal{C}_n^\mathrm{global} \right> &= \mathbf{b}_g T^n \mathbf{v}_0 \\
    &= \left(\frac12 - e\right) \mathbf{u}_1 T^n \mathbf{v}_0  + \frac12 \mathbf{u}_0 T^n \mathbf{v}_0 \\
    &= \left(\frac12 - e\right) (1 -2p -L)^n \mathbf{u}_1 \cdot \mathbf{v}_0 + \frac12 \mathbf{u}_0 \cdot \mathbf{v}_0 \\
    &= \left(\frac12 - e\right) (1 - 2h) (1 -2p -L)^n + \frac12,
\end{aligned}
\end{equation}
where the third equality follows because $\mathbf{u}_0$ and $\mathbf{u}_1$ are left eigenvectors of $T$. The global correlation is therefore expected to follow a single exponential decay with offset $\frac12$, regardless of the number of leakage states.

Note that this model heavily relies on proper randomization of the sequences of flips, and it is crucial to ensure the set of those sequences satisfies the required properties.

\subsection{Derivation of the local correlation model}
The local correlation model uses the same transition matrix $T$, but a different emission probability vector -- $\mathbf{b}_l$. To derive it, we will consider a sum of quantities
\begin{equation}
    P(r_{n-1} \oplus r_n \oplus i_{n+1} = 0, s_{n+1} | s_n) P(s_n)
\end{equation}
spanning all possible cases, where $s_n$ is the real qubit state after the $n$-th measurement. Note that only $P(s_n)$ depends on $n$, according to the assumption of Markovianity. 

Specifically, $P(s_n=g) = P(s_n=e) = \frac12 (v_n^0 + v_n^1)$, because the states $\ket{g}$ or $\ket{e}$ may be wanted or unwanted with probability $\frac12$, while $P(s_n = l_i) = v_n^{i+2}$. Consequently, the two first components of the emission probability vector $\mathbf{b}_l$ have to be equal, because we must have
\begin{equation}
    b_l^0 v_n^0 + b_l^1 v_n^1 = A_0 \frac12(v_n^0 + v_n^1) + A_1 \frac12(v_n^0 + v_n^1),
\end{equation}
where $A_0$ and $A_1$ are independent of $n$. Unfortunately, their values depend in general on all the introduced transition and assignment probabilities, because the presence of the flip does affect the transition probability and we may no longer rely on the average leakage rates -- we need to use the matrix $T_0$ directly. We may however write
\begin{equation}
\begin{aligned}
    A &= \frac12 (A_1 + A_0) \\ &= \frac12 \sum_{s_n \in \{g, e\}}\sum_{ s_{n+1}} P(r_{n-1} \oplus r_n \oplus i_{n+1} = 0, s_{n+1} | s_n),
\end{aligned}
\end{equation}
where $A$ in general depends on all other model parameters.

 We will now show that each leakage state contributes half its population to the local correlation. Consider an $i$-th leakage state $l_i$. The transition probability to another state is given by $T_0^{i+2, j}$. The presence of the flip does not change the state, so it does not affect the transition probability. Without the flip, the bit assignments for $(r_{n-1}, r_n)$ as $00$ and $11$ contribute to the local correlation, but with the flip, the assignments $10$ and $01$ contribute to it. These are all possible assignments, and the probability of obtaining them is independent of the presence of the flip, thus for leakage states
\begin{equation}
\begin{aligned}
    P(r_{n-1} \oplus r_n \oplus i_{n+1} = 0, s_{n+1} = j | s_n = l_i) =\\ = \frac12 P(s_{n+1} = j | s_n = l_i)  \\= \frac12 T_0^{i+2, j},
\end{aligned}
\end{equation}
where the factor of $\frac12$ is the probability of the flip. The total expected contribution factor to the local correlation if $s_n=l_i$ is:
\begin{equation}
    \sum_j^{k+2} \frac12 T_0^{i+2, j} = \frac12,
\end{equation}
because columns of $T_0$ sum to 1.

We may finally define the emission probability vector for local correlation 
\begin{equation}
    \mathbf{b}_l  = [A, A, \frac12, \cdots]  = (A - \frac12) [1, 1, 0, \cdots] + \frac12\mathbf{u}_0.
\end{equation}
As shown before, $\mathbf{v}_{-} = [1, -1, 0, \cdots]^T$ is a right eigenvector of $T$ with eigenvalue $(1-2p-L)$. Define additionally $\mathbf{v}_{+} = [1, 1, 0, \cdots]^T$, which allows us to write $\mathbf{v}_0 = \frac12 \mathbf{v}_{+} + (\frac12 - h) \mathbf{v}_{-}$. Thus we obtain for the expectation value of the local correlation:
\begin{equation}\label{eq:local_expectation_general}
\begin{aligned}
    \left<\mathcal{C}_n^\mathrm{local} \right> &=  \mathbf{b}_l T^n \mathbf{v}_0 \\
    &= \frac12 \mathbf{b}_l T^n \mathbf{v}_{+} + \left(\frac12 - h\right)\mathbf{b}_l T^n \mathbf{v}_{-} \\
    &= \frac12 \mathbf{b}_l T^n \mathbf{v}_{+} + \left(\frac12 - h\right) (1-2p-L)^n \mathbf{b}_l \cdot \mathbf{v}_{-} \\
    &= \frac12 \mathbf{b}_l T^n \mathbf{v}_{+},
\end{aligned}
\end{equation}

because the vector $\mathbf{b}_l$ is orthogonal to $\mathbf{v}_{-} = \mathbf{u}_1^T$ associated with eigenvalue $(1-2p-L)$. Thus we see that the state preparation error does not affect the local correlation. However, $\mathbf{b}_l$ is not guaranteed to be orthogonal to any other eigenvector. In almost all cases, the matrix $T$ will be diagonalizable and have $k$ more eigenvectors and eigenvalues, each of which may be dependent on any remaining leakage and seepage rates, including the transition probabilities from one leakage state to another. Thus, we expect the local correlation to be a decaying function of $n$, with $k$ exponential contributions. 

However, typically only one leakage state is relevant. In such a case, the matrix $T$ becomes:
\begin{equation}\label{eq:transition_one_leakage}
    T = 
    \begin{bmatrix}
        1- p - L & p & \frac{S}2  \\
        p & 1 - p - L & \frac{S}2 \\
        L & L & 1 - S \\
    \end{bmatrix},
\end{equation}
where $S$ is now total seepage rate; and gains a left eigenvector $\mathbf{u}_2 = [-\frac{L}{S}, -\frac{L}{S}, 1]$ with eigenvalue $1 - L - S$. Following Eq.~\eqref{eq:local_expectation_general} we obtain

\begin{equation}\label{eq:local_expectation_one_state}
\begin{aligned}
    \left<\mathcal{C}_m^\mathrm{local} \right> = \frac{L(A-\frac12)(1-L -S)^n + AS + \frac12L}{L + S},
\end{aligned}
\end{equation}
thus reconstructing the model from Ref.~\cite{hazra2025benchmarking}. We may then estimate the leakage rate $L$ and the seepage rate $S$ by fitting the amplitude, offset and decay rate of the exponentially decaying data. 


\subsection{Estimating the QNDness}

In case of one relevant leakage state, the local correlation allows measurement of the leakage rate. This rate can then be used to calculate the QNDness $\mathcal{Q}$ by combining the result with the global correlation decay constant estimated using Eq.~\eqref{eq:global_corr_expectation}. 

In case of more than one relevant leakage state, the total leakage rate from the computational subspace states may no longer be estimated as accurately without additional assumptions specifying certain transition rates to be $0$, because the local correlation is a sum of multiple exponential decays, whose rates are eigenvalues of $T$ with generally complicated analytical expressions. In other words, the model is no longer fully identifiable~\cite{catchpole97redundancy} -- a model with $k$ leakage states can have up to $k^2 + k + 1$ parameters, including $A$ and all the independent rates included in the matrix $T$. It then generates data following a sum of $k$ exponential decays, which can be fitted with $2k+1$ parameters: a global offset, and one scaling factor and decay rate per included exponential contribution. Thus, the model becomes non-identifiable once $k^2 + k + 1 > 2k + 1$ or already for $k \ge 2$. It is unknown to us whether the total leakage rate from the computational subspace is an identifiable parameter without additional assumptions.

In such a case, one may bound the QNDness $\mathcal{Q}$ using just the global correlation rate $2p + L$:
\begin{equation}
\begin{aligned}
    1 - 2p - L &< \mathcal{Q} < 1 - p - \frac{L}2, \\
\end{aligned}
\end{equation}
where the lower bound may be taken if leakage is the limiting factor.

\subsection{Imperfect single qubit gates}
For the purpose of the above derivation, we have assumed perfect single qubit gates. Given low errors demonstrated in the Section~\ref{subsec:single_qubit_gates} of the main text, this is a reasonable assumption, and here we demonstrate how to further test its validity. 
Suppose our single qubit $X$ rotation is not perfect and the entirety of its infidelity is due to a depolarizing error $\lambda$. In such a case we would have $\lambda = \epsilon_{X} \approx 2 \epsilon_{\sqrt{X}}$ because our $X$ gates are implemented using two $\sqrt{X}$ gates.

Since during the RILB experiment any $X$ rotation is immediately followed by a $Z$ basis measurement, we can assume its action is effectively given by a classical stochastic matrix $E_{\lambda}$:
\begin{equation}
E_{\lambda} = 
    \begin{bmatrix}
        1- \lambda & \lambda & 0 & \cdots \\
        \lambda & 1 - \lambda & 0 & \cdots\\
        0 & 0 & 1 & \\
        \vdots & \vdots & & \ddots
    \end{bmatrix}.
\end{equation}
This matrix is applied to the matrix $T$ when the $X$ rotation is present, which occurs with probability $\frac12$. The wanted state is also $\ket{g}$ or $\ket{e}$ with probability $\frac12$, and thus the new effective transition matrix becomes:
\begin{equation}
    T' = \frac14 ( T_0 +  T_1 E_{\lambda} + T_0 E_{\lambda} + T_1) =  T \left(\frac12 I + \frac12 E_{\lambda} \right),
\end{equation}
which is a matrix structurally identical to $T$, with the substitution
\begin{equation}
    p \rightarrow p + \frac{\lambda}{2} - \frac{L \lambda}{2} - p \lambda.
\end{equation}
Thus, by assuming perfect $X$ rotations we may overestimate the real $p$ -- and so underestimate the QNDness -- by an additive factor up to $\frac{\lambda}{2} \approx \epsilon_{\sqrt{X}}$, while the measured estimate for the leakage rate is unaffected.

Note that the models for global and local correlations rely on symmetry introduced by random flips. If the single qubit gate's imperfections are of a different nature, it might break that symmetry, violating the model.

We include this correction in the reported uncertainty in the QNDness. 
\subsection{QND-optimized readout vs shelved readout}

Fig.~\ref{fig:fig_sm_10}(b) shows experimental data from the RILB experiment using a QND-optimized readout configuration. This readout consists of a single rectangular pulse with a duration of \SI{240}{\nano\second}. We extract the local and global correlations for both $\mathrm{Q}_1$ and $\mathrm{Q}_2$ as described in Appendix~\ref{subsec:rilb_experiment}, and fit the data using the models from Eqs.~\eqref{eq:global_corr_expectation} and \eqref{eq:local_expectation_one_state}. Then, we calculate the QNDness $\mathcal{Q}$ by combining the results. The extracted parameters are as follows:
\begin{center}
\begin{tabular}{lcc}
\toprule
 & $\mathrm{Q}_1$ & $\mathrm{Q}_2$ \\
\midrule
$1 - \mathcal{Q}$ & \SI{6.7(0.2)e-3}{} & \SI{5.9(0.2)e-3}{} \\
$L$ & \SI{2.1(0.2)e-5}{} & \SI{6.6(0.3)e-4}{} \\
\bottomrule
\end{tabular}
\end{center}

The large relative uncertainties in the extracted leakage ($L$) rates are due to the flatness of the local correlation curves, indicating negligible leakage under these conditions.

In contrast, Fig.~\ref{fig:fig_sm_10}(c) shows RILB data obtained using a shelved readout scheme, where the qubit is excited to the second excited state before readout (parameters as used in the main text). The second $X_\pi^{12}$ after the readout pulse ideally brings the $\ket{f}$ state back the $\ket{e}$. Due to larger overall leakage, for local correlation we select either the single exponential model from Eq.~\eqref{eq:local_expectation_one_state} or a phenomenological two-exponential model by minimizing the Bayesian Information Criterion (BIC)~\cite{ding2018modelselection}.

In case of both qubits, the single exponential model explains the data better, implying the presence of one dominating leakage state. The extracted parameters are:
\begin{center}
\begin{tabular}{lcc}
\toprule
 & $\mathrm{Q}_1$ & $\mathrm{Q}_2$ \\
\midrule
$1 -\mathcal{Q}$ & \SI{1.06(0.02)e-2}{} & \SI{1.11(0.02)e-2}{} \\
$L$ & \SI{1.005(0.006)e-2}{} &  \SI{1.05(0.005)e-2}{} \\
\bottomrule
\end{tabular}
\end{center}

Comparing the two configurations highlights the main drawback of shelved readout: a roughly two orders of magnitude increase in the leakage rate. This degradation is attributed to the decay of the $\ket{f}$ state to $\ket{e}$ during the readout pulse, followed by the subsequent $X_\pi^{12}$ pulse re-exciting $\ket{e}$ to $\ket{f}$, resulting in persistent population in the $\ket{f}$ state after measurement. On the other hand, the switching rate $p$ does not contribute to QNDness, implying that the shelved readout is working as designed by trading the higher leakage for the substantial decrease in the switching rate $p$.

The comparison also yields an important insight for the RILB method: even though the global correlation exponential decays in Fig.~\ref{fig:fig_sm_10}(b) and (c) are similar, the QNDness of the shelved readout is almost entirely limited by leakage, whereas the QNDness of the QND optimized readout is almost entirely limited by the switching rate $p$ between the two computational states. As a result, the QND error differs by a factor of $2$. Ultimately, both the switching rate and leakage contribute to QNDness and are important metrics, and we demonstrate that it is possible to suppress one in favour of the other. The final choice of which readout is better depends on the application; for example for Quantum Error Correction, reducing leakage is typically much more important, because the leakage error cannot be corrected by the codes without the use of additional resources such as the Leakage Reduction Unit~\cite{marques2023LRU}.


\end{document}